\def\draftdate{November 26, 2013}

\documentclass{amsart}
\usepackage{amssymb}
\usepackage{xy}
\usepackage{url}

\ifnum\pdfoutput=0
  \newcommand\includegraphics[2][{}]{%
{[Picture of undetermined size goes here]}}
\else
  \usepackage{graphicx}
\fi

\xyoption{arrow}
\xyoption{matrix}
\xyoption{cmtip}

\def\headerstrut{\vrule height2.25ex width0em depth0em}

\def\mhdstrut{\hbox to 0pt{\hss\vrule height1.75ex width0pt depth0pt}}

\let\catsymbfont\mathcal
\newcommand{\aA}{{\catsymbfont{A}}}
\newcommand{\aB}{{\catsymbfont{B}}}
\newcommand{\aBc}{\overline\aB}

\newcommand{\aD}{{\catsymbfont{D}}}
\newcommand{\aE}{{\catsymbfont{E}}}

\newcommand{\aI}{{\catsymbfont{I}}}

\newcommand{\aM}{{\catsymbfont{M}}}

\newcommand{\aP}{{\catsymbfont{P}}}

\newcommand{\aS}{{\catsymbfont{S}}}

\newcommand{\aZ}{{\catsymbfont{Z}}}

\newcommand{\bN}{{\mathbb{N}}}

\newcommand{\bR}{{\mathbb{R}}}

\def\quickop#1{\expandafter\DeclareMathOperator\csname
#1\endcsname{#1}}

\quickop{VR}\quickop{HB}\quickop{PH}\quickop{HD}\quickop{MHD}\quickop{diam}\quickop{supp}\quickop{dis}\quickop{Eur}\quickop{sComp}\quickop{W}\quickop{Hyp}

\newcommand{\PHD}{\Phi}

\DeclareMathOperator*\median{median}

\newcommand{\mypstrut}{{\setbox0=\vbox{$p$}\vrule height0pt width0pt depth \dp0}}
\DeclareMathOperator*\dsup{sup\mypstrut}
\DeclareMathOperator*\dinf{inf\mypstrut}

\numberwithin{equation}{section}

\newtheorem{thm}[equation]{Theorem}
\newtheorem*{thm*}{Theorem}
\newtheorem{cor}[equation]{Corollary}
\newtheorem{lem}[equation]{Lemma}
\newtheorem{prop}[equation]{Proposition}
\theoremstyle{definition}
\newtheorem{defn}[equation]{Definition}
\newtheorem{notn}[equation]{Notation}
\newtheorem{alg}[equation]{Algorithm}

\theoremstyle{remark}
\newtheorem{rem}[equation]{Remark}

\begin{document}

\title[Persistent homology on metric measure spaces]
{Robust statistics, hypothesis testing, and confidence intervals
for persistent homology on metric measure spaces}

\author[Blumberg]{Andrew J. Blumberg}
\address{Department of Mathematics, University of Texas at Austin,
Austin, TX \ 78712}
\email{blumberg@math.utexas.edu}

\author[Gal]{Itamar Gal}
\address{Department of Mathematics, University of Texas at Austin,
Austin, TX \ 78712}
\email{igal@math.utexas.edu}

\author[Mandell]{Michael A. Mandell}
\address{Department of Mathematics, Indiana University
Bloomington, IN \ 47405}
\email{mmandell@indiana.edu}

\author[Pancia]{Matthew Pancia}
\address{Department of Mathematics, University of Texas at Austin,
Austin, TX \ 78712}
\email{mpancia@math.utexas.edu}

\thanks{The authors were supported in part by DARPA YFA award N66001-10-1-4043}

\date{\draftdate}

\begin{abstract}
We study distributions of persistent homology barcodes
associated to taking subsamples of a fixed size from metric measure
spaces.  We show that such distributions provide robust invariants of
metric measure spaces, and illustrate their use in hypothesis testing
and providing confidence intervals for topological data analysis.
\end{abstract}

\maketitle

\section{Introduction}

Topological data analysis assigns homological invariants to data
presented as a finite metric space (a ``point cloud'').  If we imagine
this data as measurements sampled from some abstract universal space
$X$, the structure of that space is a metric measure space, having a
notion both of distance between points and a notion of probability for
the sampling.  A standard homological approach to studying the samples
is to assign a simplicial complex and compute its homology.  The
construction of the associated simplicial complex for a point cloud
depends on a choice of scale parameter.  The insight of
``persistence'' is that one should study homological invariants that
encode change across scales; the correct scale parameter is a priori
unknown.  As such, a first approach to studying the homology of $X$
from the samples is to simply compute the persistent homology
$\PH_*(\tilde{X})$ of simplicial complexes associated to the sampled
point cloud $\tilde{X}$. 

We can gain some perspective from imagining that we could make
measurements on $X$ directly and interpret these measurements in terms
of random sample points.  With this in mind, we immediately notice
some defects with homology and persistent homology as invariants of
$X$.  While the homology of $X$ captures information about the global
topology of the metric space, the probability space structure plays no
role.  This has bearing even if we assume $X$ is a compact Riemannian
manifold and the probability measure is the volume measure for the
metric: handles which are small represent subsets of low probability
but contribute to the homology in the same way as large handles.  In
this particular kind of example, persistent homology can identify this
type of phenomenon (by encoding the scales at which homological
features exist); however, in a practical context, the metric on the
sample may be ad hoc (e.g., \cite{carlssonetalvision}) and less
closely related to the probability measure.  In this case, we could
have handles that are medium size with respect to the metric but still
low probability with respect to the measure.  Homology and persistent
homology have no mechanism for distinguishing low probability features
from high probability features.  A closely related issue is the effect
of small amounts of noise (e.g., a situation in which a fraction of
the samples are corrupted).  A small proportion of bad samples can
arbitrarily change the persistent homology.  These two kinds of
phenomena are linked, insofar as decisions about whether low
probability features are noise or not is part of data analysis.

The disconnect with the underlying probability measure presents a
significant problem when trying to adapt persistent homology to the
setting of \emph{hypothesis testing} and \emph{confidence intervals}.
Hypothesis testing involves making quantitative statements about the
probability that the persistent homology computed from a sampling from
a metric measure space is consistent with (or refutes) a hypothesis
about the actual persistent homology.  Confidence intervals provide a
language to understand the variability in estimates introduced by the
process of sampling.  Because low probability features and a small
proportion of bad samples can have a large effect on persistent
homology computations, the persistent homology groups make poor test
statistics for hypothesis testing and confidence intervals.  To obtain
useable test statistics, we need to develop invariants that better
reflect the underlying measure and are less sensitive to large
perturbation.  To be precise about this, we use the statistical notion
of robustness.

A statistical estimator is \emph{robust} when its value cannot be
arbitrarily perturbed by a constant proportion of bad samples.  For
instance, the sample mean is not robust, as a single extremely large
sample value can dominate the result.  On the other hand, the sample
median is robust.  As we discuss in Section~\ref{sec:notrobust},
persistent homology is not robust.  A small number of bad samples can
cause large changes in the persistent homology, essentially as a
reflection of the phenomenon of large metric low probability handles
(including spurious ones).

In order to handle this, we adopt a standard statistical perspective,
namely that the distribution of an estimator on some fixed finite
number of samples is an appropriate way to grapple with such
behavior.  To make this precise, we need to be able to talk about
probability distributions on homological invariants.

Using the idea of an underlying metric measure space $X$, formally the
process of sampling amounts to considering random variables on the
probability space $X^{n}=X\times \dotsb \times X$ equipped with the
product probability measure.  The $k$-th persistent homology of a size
$n$ sample is a random variable on $X^{n}$ taking values in the set
$\aB$ of finite \emph{barcodes}~\cite{zomorodiancarlsson}, where a
barcode is essentially a multiset of intervals of the form $[a,b)$.
The set $\aB$ of barcodes is equipped with a metric $d_{\aB}$, the
bottleneck metric \cite{cohensteiner}, and we show in
Section~\ref{sec:polish} that it is separable and that its completion
$\aBc$ is also a space of barcodes.  Then $\aBc$ is Polish, i.e.,
complete and separable, which makes it amenable to probability theory
(see also~\cite{polish} for similar results).  In particular, various
metrics on the set of distributions on $\aBc$ metrize weak
convergence, including the Prohorov metric $d_{Pr}$ and the
Wasserstein metric $d_{W}$.  We consider the following probability
distribution on barcode space $\aBc$ (restated in
Section~\ref{sec:distinv} as Definition~\ref{defn:PHD}).

\begin{defn}\label{defn:PH}
For a metric measure space $(X, \partial_X, \mu_X)$ and fixed $n,
k \in \bN$, define $\PHD_k^n$ to be the empirical measure induced by
the $k$th $n$-sample persistent homology, i.e., 
\[
\PHD_{k}^{n}(X, \partial_X, \mu_X) = (\PH_k)_* (\mu_X^{\otimes n}), 
\] 
the probability distribution on the set of barcodes $\aBc$ induced by
pushforward along $\PH_k$ from the product measure $\mu_X^{n}$ on
$X^n$.
\end{defn}

In other words, $\PHD_{k}^{n}$ is the probability measure on the space
of barcodes where the probability of a subset $A$ is the probability
that a size $n$ sample from $X$ has $k$-th persistent homology landing
in $A$.  Note that the pushforward makes sense since $\PH_k$ is a
continuous and hence Borel measurable function; see
Section~\ref{sec:distinv} for a discussion.

Although complicated, $\PHD_{k}^{n}(X)$ is a continuous invariant of
$X$ in the following sense.  The moduli space of metric measure spaces
admits a metric (in fact several) that combine the ideas of the
Gromov-Hausdorff distance on compact metric spaces and weak
convergence of probability measures \cite{sturm}.  We follow
\cite{greven}, and use the
\emph{Gromov-Prohorov metric}, $d_{GPr}$.  We prove the following
theorem in Section~\ref{sec:distinv} (where it is restated as
Theorem~\ref{thm:main}). 

\begin{thm}\label{thm:stability}
Let $(X, \partial_X, \mu_X)$ and $(X', \partial_{X'}, \mu_{X'})$ be compact
metric measure spaces.  Then we have the following inequality relating
the Prohorov and Gromov-Prohorov metrics:
\[
d_{Pr}(\PHD_{k}^{n}(X, \partial_X, \mu_X), \PHD_{k}^{n}(X', \partial_{X'}, \mu_{X'})) \leq n \,
d_{GPr}((X, \partial_X, \mu_X), (X', \partial_{X'}, \mu_{X'})).
\]
\end{thm}

This inequality becomes increasingly tight as the right-hand side
approaches $0$; we discuss precise estimates in
Section~\ref{sec:distinv}.  As we explain there, the fact that the
bound increases with $n$ is expected behavior: $n$ should be thought
of as a scale parameter, and increasing $n$ yields a more sensitive
invariant.  The main import of Theorem~\ref{thm:stability} is that for
fixed $n$, empirical approximations to $\PHD_k^n(X, \partial_X,
\mu_X)$ computed from subsets $S \in X$ are asymptotically convergent
as the number of samples increase.

Theorem~\ref{thm:stability} therefore validates computing
$\PHD_{k}^{n}$ in practice using 
empirical approximations, where we are given a large finite sample $S$
which we regard as drawn from $X$.  Making $S$ a metric measure space
via the subspace metric from $X$ and the empirical measure, we can
compute $\PHD_{k}^{n}(S)$ as an approximation to $\PHD_{k}^{n}(X)$.
This procedure is justified by the fact that as the sample size
increases, the empirical metric converges (in $d_{GPr}$)
to $X$; see Corollary~\ref{cor:largenum}.  In particular, this
justifies a resampling procedure to approximate $\PHD_k^n$ by subsampling
from a large sample of size $N$.  (We can also approximate using more
sophisticated resampling methodology, a topic we study in future work.)

Moreover, as a consequence of the continuity implied by the previous
theorem, we can use $\PHD_{k}^{n}$ to develop robust statistics: If we
change $X$ by adjusting the metric arbitrarily on $\epsilon$
probability mass to produce $X'$, then the Gromov-Prohorov distance
satisfies $d_{GPr}(X,X')\leq \epsilon$.

A difficulty with applying $\PHD_{k}^{n}$ is that it can be hard to
interpret or summarize the information contained in a distribution of
barcodes, unlike distributions of numbers for which there are various
moments (e.g., the mean and the variance) which provide concise
summaries of the distribution.  One approach is to develop
``topological summarizations'' of distributions of barcodes; a version
of this using Frechet means is explored in~\cite{frechet}.  Another
possibility is to embed the space of barcodes in a more tractable
function space~\cite{bubeniklandscapes}.  In this paper, we instead
consider cruder invariants which take values in $\bR$.  One such
invariant is the distance with respect to a reference distribution on
barcodes $\aP$, chosen to represent a hypothesis about the persistent
homology of $X$.

\begin{defn}\label{defn:HD}
Let $(X, \partial_X, \mu_X)$ be a compact metric measure space and let
$\aP$ be a fixed reference distribution on $\aBc$.   Fix
$k,n \in \bN$.  Define the homological distance on $X$ relative to
$\aP$ to be
\[
\HD_k^n((X,\partial_X, \mu_X), \aP) =
d_{Pr}(\PHD_{k}^{n}(X,\partial_X, \mu_X), \aP).  
\]
\end{defn}

We also consider a robust statistic $\MHD_{k}^{n}$ related to 
$\HD_{k}^{n}$ without first computing the distribution
$\PHD_{k}^{n}$.  To construct $\MHD_{k}^{n}$, we start with a
reference barcode and compute the median distance to the barcodes of
subsamples.

\begin{defn}\label{defn:MHD}
Let $(X, \partial_X, \mu_X)$ be a compact metric measure space and fix
a reference barcode $B \in \aBc$.   Fix $k,m \in \bN$.  Let $\aD$
denote the distribution on $\bR$ induced by applying $d_{\aB}(B,-)$ to
the barcode distribution $\PHD_{k}^{n}(X,\partial_X, \mu_X)$.  Define
the median homological distance relative to $B$ to be 
\[
\MHD_k^n ((X,\partial_X,\mu_X), B)
= \median(\aD).
\]
\end{defn}

\begin{rem}
The appearance of reference barcodes and distributions in the
invariants above raises the question of where one obtains these
quantities.  As we illustrate below in Section~\ref{sec:examples}, a
common source of reference point is simply an a priori hypothesis
about the data that we wish to test.  The Frechet
mean~\cite{frechet} of a collection of samples provides a more
principled approach to producing such reference points.
\end{rem}

The use of the median rather than the mean in the preceding definition
ensures that we compute a robust statistic, in the following sense.

\begin{defn}\label{defn:robust}
Let $f$ be a function from finite metric spaces to a metric space
$(B,d)$. We say that $f$ is robust with robustness coefficient $r>0$ if for
any non-empty finite metric space $(X,\partial)$, there exists a bound
$\delta$ such that for any isometric embedding of $X$ into a finite
metric space $(X',\partial')$, $|X'|/|X|<1+r$ implies
$d(f(X,\partial),f(X',\partial'))<\delta$, where $|X|$ denotes the
number of elements of $X$.
\end{defn}

For example, under the analogous definition on finite multi-subsets of $\bR$
(in place of finite metric spaces), median defines a function to $\bR$
that is robust with robustness coefficient $1-\epsilon$ for any $\epsilon$
since expanding a multi-subset $X$ to a larger one $X'$ with fewer
than twice as many elements will not change the median by more than
the diameter of $X$.  Similarly, for a finite metric space $X$,
expanding $X$ to $X'$, the proportion of $n$-element samples of $X'$
which are samples of $X$ is $(|X|/|X'|)^{n}$; when this number is more
than $1/2$, the median value of any function $f$ on the set of
$n$-element samples of $X'$ is then bounded by the values of $f$ on
$n$-element samples of $X$.  Since $(N/(N+rN))^{n}>1/2$ for $r<2^{1/n}-1$,
any such function $f$ will be robust with robustness coefficient $r$ satisfying
this bound, and in particular for $r=(\ln 2)/n$.  

\begin{thm}\label{thm:AHDrobust}
For any $n,k,\aP$, the function $\MHD_k^n(-,\aP)$ from finite metric
spaces (with the uniform probability measure) to $\bR$ is robust with
robustness coefficient $> (\ln 2)/n$.
\end{thm}

The function $\PHD_{k}^{n}$ from finite metric spaces to distributions
on $\aBc$ and the function $\HD^{n}_{k}$ from finite metric spaces to
$\bR$ are robust for any robustness coefficient for trivial reasons
since the Gromov-Prohorov metric is bounded.  However, for these
functions we can give explicit uniform estimates for how much these functions
change when expanding $X$ to $X'$ just based on $|X'|/|X|$.  We
introduce the following notion of uniform robustness which is strictly
stronger than the notion of robustness.

\begin{defn}\label{defn:uniformrobust}
Let $f$ be a function from finite metric spaces to a metric space
$(B,d)$. We say that $f$ is uniformly robust with robustness
coefficient $r>0$ and estimate bound $\delta$ if for
any non-empty finite metric space $(X,\partial)$ and any isometric
embedding of $(X,\partial)$ into a finite metric space 
$(X',\partial')$, $|X'|/|X|<1+r$ implies
$d(f(X,\partial),f(X',\partial'))<\delta$.
\end{defn}

Uniform robustness gives a uniform estimate on the change in the
function from expanding the finite metric space.  For example, the
median function does not satisfy the analogous notion of uniform
robustness for functions on finite multi-subsets of $\bR$. We show in
Section~\ref{sec:distinv} that $\PHD_{k}^{n}$ and $\HD^{n}_{k}$ satisfy
this stronger notion of uniform robustness.

\begin{thm}\label{thm:HDunifrob}
For fixed $n,k$, $\PHD_{k}^{n}$ is uniformly robust with robustness
coefficient $r$ and estimate bound $nr/(1+r)$ for any $r$.  For fixed
$n,k,\aP$, $\HD^{n}_{k}(-,\aP)$ is uniformly robust with robustness
coefficient $r$ and estimate bound $nr/(1+r)$ for any $r$.
\end{thm}

As with $\PHD_{k}^{n}$ itself, the law of large numbers and the
convergence implied by Theorem~\ref{thm:stability} tells us that given
a sufficiently large finite sample $S \subset M$, we can approximate
$\HD_k^n$ and $\MHD_k^n$ of the metric measure space $M$ in a robust
fashion from the persistent homology computations on $S$.  (See
Corollaries~\ref{cor:largenum},~\ref{lem:largenumHD},
and~\ref{lem:largenumAHD} below.)

In light of the results on robustness and asymptotic convergence,
$\HD_k^n$, $\MHD_k^n$, and $\PHD_k^n$ (as well as various distributional
invariants associated to $\PHD_k^n$) provide good test statistics for
hypothesis testing.  Furthermore, one of the benefits of the
numerical statistics $\HD_k^n$ and $\MHD_k^n$ is that we can use
standard techniques to obtain confidence intervals, which provide a
means for understanding the reliability of analyses of data sets.  We
discuss hypothesis testing and the construction of confidence
intervals in Section~\ref{sec:intervals}, and explore examples in
Sections~\ref{sec:examples} and~\ref{sec:notklein}.  In this paper we
primarily focus on analytic methods and Monte Carlo simulation for
obtaining confidence intervals; however, these statistics are
well-suited for the construction of resampling confidence intervals.
In a follow-up paper~\cite{BMmm2} we establish the asymptotic
consistency of the bootstrap for $\HD_k^n$ and $\MHD_k^n$.

We regard this paper as a step towards providing a foundation for the
integration of standard statistical methodology into computational
algebraic topology.  Our goal is to provide tools for practical use in
topological data analysis.

\subsection*{Related work}  We have developed an approach to using statistical
tools to study persistent homological invariants for metric measure
spaces accessed through finite samples.  There are a number of related
approaches to studying the statistical properties of persistent
homological estimators; we quickly survey this work.

Bubenik~\cite{bubeniklandscapes} develops statistical inference via an
embedding into function spaces called ``persistence landscapes'', and
with various co-authors in~\cite{bubenikcortical, bubenikstatmorse}
studies an approach using Morse theory (and hence taking advantage of
the ambient metric space for smoothing).  The work of
Harer, Mileyko, and Mukerjee in~\cite{polish} parallels the development in
Section~\ref{sec:polish} and introduces probability measures on
barcode space, and these ideas are developed further (with Turner) in
the context of Frechet means as ways of summarizing barcode
distributions in~\cite{frechet}.

In another direction, there has been a fair amount of work on the
topological features of random simplicial complexes and noise due to
Kahle~\cite{kahle, kahlemeckes} as well as Adler, Bobrowski, Borman,
Subag, and Weinberger~\cite{adler1, adler2,
adler3}.  This work is essential for understanding what persistent
homological ``null hypotheses'' look like, and adapted to our setting
should inform our statistical inference procedures.

Finally, there has also been a lot of excellent work arising on
studying robustness in the context of understanding distances to
measures for point clouds.  This approach was introduced
by Chazal, Cohen-Steiner, and Merigot in~\cite{measuredist1}, and was
further developed by Caillerie, Chazal, Dedecker, and Michel
in~\cite{measuredist2}.  The basic idea is that the distribution of
distances to a point cloud is a robust invariant of the point cloud;
indeed, this is closely related to the $n=2$ case of our central
invariant.  Since preservation of explicit distances is a goal of this
approach, it is more closely related to rigid geometric inference (and
manifold learning) than purely topological inference, as in our
homological approach.

\subsection*{Acknowledgments}
The authors would like to thank Gunnar Carlsson and Michael Lesnick
for useful comments, Rachel Ward for comments on a prior draft, and
Olena Blumberg for help with background and for assistance with the
analysis of the tightness of the main theorem.  We
would also like to thank the Institute for Mathematics and its
Applications for hospitality while revising this paper.


\subsection*{Outline} The paper is organized as follows.  In
Section~\ref{sec:background}, we provide a rapid review of the
necessary background on simplicial complexes, persistent homology, and
metric measure spaces.  In Section~\ref{sec:polish}, we study the
space of barcodes, establishing foundations needed to work with
distributions of barcodes.  In Section~\ref{sec:notrobust}, we discuss
the robustness of persistent homology.  In Section~\ref{sec:distinv},
we study the properties of $\PHD_{k}^{n}$, $\MHD_k^n$, and $\HD_k^n$
and prove Theorem~\ref{thm:stability}.  We discuss hypothesis testing
and confidence intervals in Section~\ref{sec:intervals}, which we
illustrate with synthetic examples in Section~\ref{sec:examples}.
Section~\ref{sec:notklein} applies these ideas to the analysis of the
natural images data in \cite{carlssonetalvision}.

\section{Background}\label{sec:background}

In this section we provide background for the framework for
topological data analysis we study in this paper.  We focus on an
approach which accesses the ambient metric measure space
$(X, \partial_X, \mu_X)$ only through finite samples, i.e., point
clouds.  
\subsection{Simplicial complexes associated to point clouds}

A standard approach in computational algebraic topology proceeds by
assigning a simplicial complex (which usually also depends on a scale
parameter $\epsilon$) to a finite metric space $(X, \partial)$.
Recall that a simplicial complex is a combinatorial model of a
topological space, defined as a collection of nonempty finite sets
$\aZ$ such that for any set $Z \in \aZ$, every nonempty subset of $Z$
is also in $\aZ$.  Associated to such a simplicial complex is the
``geometric realization'', which is formed by gluing standard
simplices of dimension $|Z|-1$ via the subset relations. (The standard
$n$-simplex has $n+1$ vertexes.) The most basic and widely used
construction of a simplicial complex associated to a point cloud is
the Vietoris-Rips complex:

\begin{defn}
For $\epsilon \in \bR$, $\epsilon \geq 0$, the Vietoris-Rips complex
$\VR_\epsilon(X)$ is the simplicial complex with vertex set $X$ such
that $[v_0, v_1, \ldots, v_n]$ is an $n$-simplex when for each pair
$v_i, v_j$, the distance $\partial(v_i, v_j) \leq \epsilon$. 
\end{defn}

The Vietoris-Rips complex is determined by its
$1$-skeleton.  The construction is functorial in the sense that for a
continuous map $f \colon X \to Y$ with Lipshitz constant $\kappa$ and
for $\epsilon \leq \epsilon'$, there is a commutative diagram
\begin{equation}\label{eq:funct}
\begin{gathered}
\xymatrix{
\VR_\epsilon(X) \ar[r] \ar[d] & \VR_{\kappa \epsilon}(Y) \ar[d] \\
\VR_{\epsilon'}(X) \ar[r] & \VR_{\kappa \epsilon'}(Y). \\
}
\end{gathered}
\end{equation}

The Vietoris-Rips complex is easy to compute, in the sense that is
straightforward to determine when a simplex is in the complex.  More
closely related to classical constructions in algebraic topology is
the Cech complex.

\begin{defn}
For $\epsilon \in \bR$, $\epsilon \geq 0$, the Cech complex
$C_{\epsilon}(X)$ is the simplicial complex with vertex set $X$ such
that $[v_0, v_1, \ldots, v_n]$ is an $n$-simplex when the intersection
\[
\bigcap_{0 \leq i \leq n} B_{\frac{\epsilon}{2}}(v_i)
\]
is non-empty, where here $B_r(x)$ denotes the $r$-ball around $x$.
\end{defn}

The Cech complex has analogous functoriality properties to the
Vietoris-Rips complex.  The Cech
complex associated to a cover of a paracompact topological space
satisfies the nerve lemma: if the cover consists of contractible
spaces such all finite intersections are contractible or empty, the
resulting simplicial complex is homotopy equivalent to the original
space.

\begin{rem}\label{rem:witness}
Both the Vietoris-Rips complex and the Cech complex can be
unmanageably large; e.g., for a set of points $Y = \{y_{1},
y_{2}, \ldots, y_{n}\}$ such that
$\partial(y_{i},y_{j}) \leq \epsilon$, every subset of $Y$ specifies a
simplex of the Vietoris-Rips complex.  As a consequence, it is often
very useful to define complexes with the vertices restricted to a
small set of landmark points; the weak witness complex is perhaps the
best example of such a simplicial complex \cite{desilvacarlsson}.  We
discuss this construction further in Section~\ref{sec:notklein}, as it
is important in the applications.
\end{rem}

The theory we develop in this paper is relatively insensitive to the
specific details of the construction of a simplicial complex
associated to a finite metric space (and scale parameter).  For
reasons that will become evident when we discuss persistence in
Subsection~\ref{sec:pers} below, the main thing we require is a
procedure for assigning a complex to $((M, \partial), \epsilon)$ that
is functorial in the vertical maps of diagram~\eqref{eq:funct} for
$\kappa =1$. 

\subsection{Homological invariants of point clouds}

In light of the previous subsection, given a metric space
$(X,\partial)$, one defines the homology at the feature scale
$\epsilon$ to be the homology of a simplicial complex associated to
$(X,\partial)$; e.g., $H_* (\VR_{\epsilon}(X))$ or
$H_{*}(C_{\epsilon}(X))$.  This latter definition is supported by the
following essential consistency result, which is in line with the
general philosophy that we are studying an underlying continuous
geometric object via finite sets of samples.

\begin{thm}[Niyogi-Smale-Weinberger \cite{niygoismaleweinberger}]
Let $(M,\partial)$ be a compact Riemannian manifold equipped with an
isometric embedding $\gamma \colon M \to \bR^n$, and let $X \subset M$
be a finite independent identically-distributed sample drawn according
to the volume measure on $M$.  Then for any $p \in (0,1)$, there are
constants $\delta$ (which depends on the curvature of $M$ and the
embedding $\gamma$) and $N_{\delta, p}$ such that if $\epsilon
< \delta$ and $|X| > N_{\delta,p}$ then the probability that
$H_*(C_\epsilon(X)) \cong H_*(M)$ is an isomorphism is $> p$.
\end{thm}

In fact, Niyogi, Smale, and Weinberger prove an effective version of
the previous result, in the sense that there are explicit numerical
bounds dependent on $p$ and a ``condition number'' which incorporates
data about the curvature of $M$ and the twisting of the embedding
$\gamma$.

Work by Latschev provides an equivalent result for $\VR_{\epsilon}(X)$,
with somewhat worse bounds, defined in terms of the injectivity radius
of $M$ \cite{latschev}.  Alternatively, one can show that in the limit
$\VR_{\epsilon}(X)$ captures the homotopy type of the underlying
manifold using the fact that there are inclusions
\[
C_\epsilon(X) \subseteq \VR_{\epsilon}(X) \subseteq C_{2\epsilon}(X).
\]

While reassuring, an unsatisfactory aspect of the preceding results is
the dependence on a priori knowledge of the feature scale $\epsilon$
and the details of the intrinsic curvature of $M$ and the nature of
the embedding.  A convenient way to handle the fact that it is often
hard to know a good choice of $\epsilon$ at the outset is to consider
multi-scale homological invariants that encode the way homology
changes as $\epsilon$ varies.  This leads us to the notion of
persistent homology.

\subsection{Persistent homology}\label{sec:pers}

Persistent homology arose more or less simultaneously and
independently in work of Robins~\cite{robins}, Frosini and Ferri and
collaborators~\cite{italian1,italian2}, and Edelsbrunner and
collaborators~\cite{EdelsbrunnerLetscherZomorodian}.  See the
excellent survey of Edelsbrunner and Harer~\cite{EHsurvey} for a more
expansive discussion of the history and development of these ideas.
The efficient algorithms and the algebraic presentation we apply
herein is due to~\cite{EdelsbrunnerLetscherZomorodian}
and~\cite{zomorodiancarlsson}.

Given a diagram of simplicial complexes indexed on $\bR$, i.e., a
complex $X_s$ for each $s \in \bR$ and maps $X_s \to X_{s'}$ for
$s \leq s'$, there are natural maps $H_*(X_s) \to H_*(X_{s'})$ induced
by functoriality.

We say that a class $\alpha \in H_p(X_i)$ is \emph{born} at time $i$ if it is
not in the image of $H_k(X_j)$ for $j < i$, and we say a class $\alpha
\in H_k(X_i)$ \emph{dies} at time $i$ if the image of
$\alpha$ is 0 in $H_k(X_j)$ for $j \geq i$.  This information about
the homology can be packaged up into an algebraic object:

\begin{defn}
Let $\{X_i\}$ be a diagram of simplicial complexes indexed on $\bR$.
The $p$th persistent $k$th homology group of $X_i$ is defined to be 
\[
H_{k,p}(X_i) = Z^i_k / (B^{i+p}_k \cap Z^i_k),
\]
where $Z$ and $B$ denote the cycle and boundary groups respectively.
Alternatively, $H_{k,p}(X_i)$ is the image of the natural map 
\[
H_k(X_i) \to H_k(X_{i+p}).
\]
\end{defn}

Barcodes provide a convenient reformulation of information from
persistent homology.  Although we will work over a field and in the
presence of suitable finiteness hypotheses which are satisfied in our
motivating examples, recent work makes it clear that this restriction
could be weakened~\cite{ChazalGen,BubenikScott}.  We assume that the values
$H_*(X_i)$ change only at a countable discrete subset of $\bR$, so
that by reindexing we have a direct system
\[
X_0 \to X_1 \to \cdots \to X_n \to \cdots,
\]
the direct system of simplicial complexes stabilizes at a finite stage
and all homology groups are finitely-generated.  Then a basic
classification result of Zomorodian-Carlsson~\cite{zomorodiancarlsson}
describes the persistent homology in terms of a barcode, a multiset of
non-empty intervals of the form $[a,b)\subset \bR$.  An interval in
the barcode indicates the birth and death of a specific homological
feature.  For reasons we explain below, the barcodes appearing in our
context will always have finite length intervals.

The Vietoris-Rips (or Cech) complexes associated to a point cloud
$(X,\partial_X)$ fit into this context by looking at a sequence of
varying values of $\epsilon$: 
\[
\VR_{\epsilon_1}(X) \to \VR_{\epsilon_2}(X) \to \cdots.
\]
We can do this in several ways, for example, using the fact that the
Vietoris-Rips complex changes only at discrete points $\{\epsilon_i\}$
and stabilizes for sufficiently large $\epsilon$, or just choosing and
fixing a finite sequence $\epsilon_{i}$ independently of $X$.  The
theory we present below makes sense for either of these choices, and
we use the following notation.

\begin{notn}
Let $(X, \partial_X)$ be a finite metric space.  For $k \in \bN$,
denote the
persistent homology of $X$ as 
\[
\PH_k((X, \partial_X)) = \PH_{k,p}(\{VR_{\epsilon_{(-)}}(X)\})
\]
for some chosen sequence $0<\epsilon_{1}< \epsilon_{2} <\dotsb  $
and $p\geq 0$.
\end{notn}

More generally, we can make analogous definitions for any functor
\[
\Psi \colon \aM \times \bR_{> 0} \to \sComp,
\]
where $\aM$ is the category of finite metric spaces and metric maps
and $\sComp$ denotes the category of simplicial complexes.  We will
call such a $\Psi$ ``good'' when the homology changes for only
finitely values in $\bR$.  In this case, we can choose the directed
system of values of $\epsilon_{i}$ to contain these transition values.

We note that for large values of the parameter $\epsilon$,
$VR_{\epsilon}(X)$ will be contractible.  Therefore, if we use the
reduced homology group in dimension 0, we get $H_{k}(VR_{\epsilon})=0$
for all $k$ for large $\epsilon$.  The barcodes associated to these
persistent homologies therefore have only finite length bars.  For
convenience in computation, we typically cut off $\epsilon$ at a
moderately high value before this breakdown occurs.  The result is a
truncation of the barcode to the cut-off point.

\subsection{Gromov-Hausdorff stability and the bottleneck metric}

By work of Gromov, the set of isometry classes of compact metric spaces
admits a useful metric structure, the Gromov-Hausdorff metric. For a
pair of finite metric spaces $(X_{1}, \partial_{1})$ and $(X_{2},
\partial_{2})$, the Gromov-Hausdorff distance is defined as follows:
For a compact metric space $(Z, \partial)$ and closed subsets $A,B
\subset Z$, the Hausdorff distance is defined to be
\[
d^{Z}_{H}(A,B) = \max (\dsup_{a \in A} \dinf_{b \in B} \partial(a,b),
\dsup_{b \in B} \dinf_{a \in A} \partial(a,b)).
\]
One then defines the Gromov-Hausdorff distance between $X_{1}$ and
$X_{2}$ to be
\[
d_{GH}(X_{1},X_{2}) = \inf_{Z,\gamma_{1},\gamma_{2}}
d^{Z}_{H}(X_{1},X_{2}),
\]
where here $\gamma_{1} \colon X_{1} \to Z$ and $\gamma_{2} \colon
X_{2} \to Z$ are isometric embeddings.

Since the topological invariants we are studying ultimately arise from
finite metric spaces, a natural question to consider is the degree to
which point clouds that are close in the Gromov-Hausdorff metric have
similar homological invariants.  This question does not in general
have a good answer in the setting of the homology of the point cloud,
but in the context of persistent homology, Chazal, et
al. \cite[3.1]{chazal} provide a seminal theorem
in this direction that we review as
Theorem~\ref{thm:GHcompare} below.  

The statement of Theorem~\ref{thm:GHcompare} involves a metric on the
set of barcodes called the bottleneck distance and defined as follows.  Recall that a barcode
$\{I_\alpha\}$ is a multiset of non-empty intervals.  Given two non-empty intervals $I_1 =
[a_1, b_1)$ and $I_2= [a_2, b_2)$, define the distance between them to
be 
\[
d_{\infty}(I_1,I_2) = ||(a_{1},b_{1})-(a_{2},b_{2})||_{\infty}=\max
(|a_1 - a_2|, |b_1 - b_2|).
\]
We also make the convention
\[
d_{\infty}([a, b),\emptyset) = |b -  a|/2
\]
for $b>a$ and $d_{\infty}(\emptyset,\emptyset)=0$.
For the purposes of the following definition, we define a
matching between two barcodes $B_{1}=\{I_{\alpha}\}$ and
$B_{2}=\{J_{\beta}\}$ to be a multi-subset $C$ of the underlying set
of 
\[
(B_{1}\cup \{\emptyset\})\times (B_{2}\cup \{\emptyset\})
\]
such that $C$ does not contain $(\emptyset,\emptyset)$ and each
element $I_{\alpha}$ of $B_{1}$ occurs as the first coordinate of an
element of $C$ exactly the number of times (counted with multiplicity)
of its multiplicity in $B_{1}$, and likewise for every element of
$B_{2}$.  We get a more intuitive but less convenient description of a
matching using the decomposition of $(B_{1}\cup \{\emptyset\})\times
(B_{2}\cup \{\emptyset\})$ into its evident four pieces: The basic
data of $C$ consists of multi-subsets $A_{1}\subset
B_{1}$ and $A_{2}\subset B_{2}$ together with a bijection (properly
accounting for multiplicities) $\gamma\colon A_{1}\to A_{2}$; $C$ is then
the (disjoint) union of the graph of $\gamma$ viewed as a multi-subset of
$B_{1}\times B_{2}$, the multi-subset $(B_{1}-A_{1})\times
\{\emptyset\}$ of $B_{1}\times \{\emptyset\}$, and the multi-subset
$\{\emptyset\}\times (B_{2}-A_{2})$ of $\{\emptyset\}\times B_{2}$.
With this terminology, we can define the bottleneck distance.

\begin{defn}
The bottleneck distance between barcodes $B_{1}=\{I_\alpha\}$ and
$B_{2}=\{J_{\beta}\}$ is
\[
d_{\aB}(B_{1}, B_{2}) = \dinf_{C} \dsup_{(I,J)\in C} d_{\infty}(I,J),
\]
where $C$ varies over all matchings between $B_{1}$ and $B_{2}$.
\end{defn}

Although expressed slightly differently, this agrees with the
bottleneck metric as defined in \cite[\S3.1]{cohensteiner} and
\cite[\S2.2]{chazal}.
On the set of barcodes $\aB$ with finitely many finite
length intervals, $d_{\aB}$ is obviously a metric.  
More generally, for any
$p>0$, one can consider 
the $\ell^{p}$ version of this metric,
\[
d_{B,p}(B_{1},B_{2})=\inf_{C} \bigl(\sum_{(I,J)\in C} d_{\infty}(I,J)^{p}\bigr)^{1/p}. 
\]
For simplicity, we focus on $d_{\aB}$ in this paper, but
analogues of our main theorems apply to these variant metrics as
well. 

We have the following essential stability theorem:

\begin{thm}[Chazal, et. al.~{\protect\cite[3.1]{chazal}}]
\label{thm:GHcompare}
For each $k$, we have the bound
\[
d_{\aB}(\PH_k(X), \PH_k(Y)) \leq d_{GH}(X,Y).
\]
\end{thm}

Note that truncating barcodes (i.e., truncating each persistent
interval) is a Lipshitz map $\aB\to\aB$ with Lipshitz constant $1$, so
the bound above still holds when we use a large parameter cut-off in
defining $\PH_{k}$. 

\begin{rem}
The space of barcodes admits other metrics that are finer than the
bottleneck metric for which versions of the stability theorem also hold;
these can be useful in practical situations.  Notably, the
papers~\cite{polish, lpstab, frechet} study and apply a family of
Wasserstein (mass transportation) metrics on barcode space.  We
believe that our results can be extended to this setting.
\end{rem}

\subsection{Metric measure spaces and the Gromov-Prohorov distance}

To establish more robust convergence results, we work with
suitable metrics on the set of compact metric measure spaces.
Specifically, following \cite{greven, memoli,sturm} we use the idea
of the Gromov-Hausdorff metric to extend certain standard metrics on
distributions (on a fixed metric measure space) to a metric on the set
of all compact metric measure spaces.

A basic metric of this kind is the Gromov-Prohorov
metric \cite{greven}.  (For the following formulas, see Section~5 of
\cite{greven} and its references.)  This metric is defined in terms of
the standard Prohorov metric $d_{Pr}$ (metrizing weak convergence of 
probability distributions on separable metric spaces).  First, recall
that for measures $\mu_1$ and $\mu_2$ on a metric space $Z$, the
Prohorov metric is defined as
\[
d_{Pr}(\mu_1, \mu_2) = \inf \{\epsilon > 0 \mid \mu_1(A)
\leq \mu_2(B_{\epsilon}(A)) + \epsilon\},
\]
where $A \subset Z$ varies over all closed sets and $B_{\epsilon}(A)$
is the set of points $z$ such that $d_Z(z,a) < \epsilon$ for some
$a \in A$.  Then the Gromov-Prohorov metric is defined as 
\[
d_{GPr}((X,\partial_X, \mu_X),(Y,\partial_Y, \mu_Y)) = \inf_{(\phi_X,
  \phi_Y, Z)} d_{Pr}^{(Z,\partial_Z)}((\phi_X)_* \mu_X, (\phi_Y)_*
\mu_Y),
\]
where the $\inf$ is computed over all isometric embeddings
$\phi_X \colon X \to Z$ and $\phi_Y \colon Y \to Z$ into a target
metric space $(Z,\partial_Z)$.


It is very convenient to reformulate both the Gromov-Hausdorff and
Gromov-Prohorov distances in terms of relations.  For sets $X$ and
$Y$, a relation $R \subset X \times Y$ is a correspondence if for each
$x \in X$ there exists at least one $y \in Y$ such that $(x,y) \in R$
and for each $y' \in Y$ there exists at least one $x' \in X$
such that $(x',y') \in R$.  For a relation $R$ on metric spaces
$(X,\partial_X)$ and $(Y,\partial_Y)$, we define the distortion as
\[
\dis(R) = \sup_{(x,y),(x',y') \in R} |\partial_X(x,x') - \partial_Y(y,y')|.
\]

The Gromov-Hausdorff distance can be expressed
as 
\[
d_{GH}((X, \partial_X), (Y, \partial_Y)) = \frac{1}{2} \inf_R \dis(R),
\]
where we are taking the infimum over all correspondences $R\subset X\times Y$.

Similarly, we can reformulate the Prohorov metric as follows.  Given
two measures $\mu_1$ and $\mu_2$ on a metric space $X$, let a
coupling of $\mu_1$ and $\mu_2$ be a measure $\psi$ on $X \times X$
(with the product metric) such that $\psi(X \times -) = \mu_2$ and
$\psi(- \times X) = \mu_1$.  Then we have
\[
d_{Pr}(\mu_1, \mu_2)
= \inf_\psi \inf \{ \epsilon>0\mid \psi \left\{(x,x') \in X \times X \,
| \, \partial_X(x,x') \geq \epsilon \right\} \leq \epsilon\}.
\]

This characterization of the Prohorov metric turns out to be useful
when working with the Gromov-Prohorov metric in light of the (trivial)
observation that if $d_{GPr}((X, \partial_X, \mu_X),
(Y, \partial_Y, \mu_Y)) < \epsilon$ then there exists a metric space
$Z$ and embeddings $\iota_1 \colon X \to Z$ and $\iota_2 \colon Y \to
Z$ such that $d_{Pr}((\iota_1)_* \mu_X, (\iota_2)_* \mu_Y) <  \epsilon$.

\section{Probability measures on the space of barcodes}\label{sec:polish}

This section introduces the spaces of barcodes $\aB_{N}$ and $\aBc$
used in the distributional invariants $\PHD_{k}^{n}$ of
Definition~\ref{defn:PH}.  These spaces are complete and separable
under the bottleneck metric.  This implies in
particular that the Prohorov metric on the set of probability measures
in $\aB_{N}$ or $\aBc$ metrizes convergence in probability, which
justifies the perspective in the stability theorem~\ref{thm:stability}
and the definition of the invariants $\HD^{n}_{k}(-,\aP)$ in
Definition~\ref{defn:HD}.  

A barcode is by definition a multi-set of intervals, in our case of
the form $[a,b)$ for $0\leq a<b<\infty$.  The set $\aI$ of all
intervals of this form is of course in bijective correspondence with a
subset of $\bR^{2}$. A multi-set $A$ of intervals is a multi-subset of
$\aI$, which concretely is a function from $\aI$ to the natural
numbers $\bN=\{0,1,2,3,\dotsc\}$ which counts the number of multiples
of each interval in $A$.  We denote by $|A|$ the cardinality of $A$,
which we define as the sum of the values of the function $\aI\to \bN$
specified by $A$ (if finite, or countably or uncountably infinite, if
not).  The space $\aB$ of barcodes of the introduction is 
the set of multi-sets of intervals $A$ such that $|A|<\infty$.  We
have the following important subsets of $\aB$.

\begin{defn}\label{defn:aBN}
For $N\geq 0$, let $\aB_{N}$ denote the set of multi-sets of
intervals (in $\aI$) $A$ with $|A|\leq N$.
\end{defn}

The main result on $\aB_{N}$ is the following theorem, proved below.
(Similar results can also be found in~\cite{polish}.) 

\begin{thm}\label{thm:aBN}
For each $N\geq 0$, $\aB_{N}$ is complete and separable under the
bottleneck metric.
\end{thm}

Since the homology $H_{k}$ (with any coefficients) of any complex with
$n$ vertices can have rank at most $\binom{n}{k+1}$, our persistent
homology barcodes will always land in one of the $\aB_{N}$, with $N$
depending just on the size of the samples.  As we let the size of the
samples increase, $N$ may increase, and so it is convenient to have a
target independent of the number of samples.  The space $\aB=\bigcup
\aB_{N}$ is clearly not complete under the bottleneck metric (consider
a sequence of barcodes $\{X_n\}$ such that $X_n$ is produced from
$X_{n-1}$ by adding a bar $(0,\frac{1}{n})$), so we introduce the
following space of barcodes $\aBc$.

\begin{defn}\label{defn:aBc}
Let $\aBc$ be the space of multi-sets $A$ of intervals (in $\aI$) with the
property that for every $\epsilon>0$, the
multi-subset of $A$ of those intervals of length more than $\epsilon$
has finite cardinality.
\end{defn}

Clearly barcodes in $\aBc$ have at most countable cardinality, and the
bottleneck metric extends to a pseudo-metric
$d_{\aB}\colon \aBc\times \aBc\to \bR$.  The following lemma shows it
is a metric.

\begin{lem}\label{lem:zeromet}
For $X,Y\in \aBc$, $d_{\aB}(X,Y)=0$ only if $X=Y$.
\end{lem}

\begin{proof}
Let $X,Y\in \aBc$ with $d_{\aB}(X,Y)=0$ and assume without loss of
generality that $X$ is not in $\aB_{N}$ for any $N$.  Then the
possible distinct lengths of intervals in $X$ or $Y$ form a countable set
$\ell_{0}>\ell_{1}>\dotsb $.  Let $X_{i}$ and $Y_{i}$ denote the
multisubsets of $X$ and $Y$ consisting of the intervals of length exactly
$\ell_{i}$.  Let $\epsilon_{0}<(\ell_{0}-\ell_{1})/2$
and in general let 
\[
\epsilon_{i}<\min(\epsilon_{0},\dotsc,\epsilon_{i-1},(\ell_{i}-\ell_{i+1})/2)
\]
(with each $\epsilon_{i}>0$). For any $n$ and any
$0<\epsilon<\epsilon_{n}$, any matching $C$ of $X$ and $Y$ with 
\[
d_{C}(X,Y) = \sup_{(I,J)\in C}d_{\infty}(I,J) < \epsilon 
\]
must induce a bijection between $X_{i}$ and $Y_{i}$ for all $i\leq
n$; moreover, if $C_{i}$ denotes the restriction of $C$ to a matching
of $X_{i}$ and $Y_{i}$, $d_{C_{i}}(X_{i},Y_{i})<\epsilon$. Letting
$\epsilon$ go to zero, we see that $X_{i}=Y_{i}$ for all $i$ and that
$X=Y$.  
\end{proof}
 
Lemma~\ref{lem:zeromet} implies that $d_{\aB}$ extends to a metric on
$\aBc$.  We prove the following theorem.

\begin{thm}\label{thm:aBc}
$\aBc$ is the completion of $\aB=\bigcup \aB_{N}$ in the bottleneck metric.
In particular $\aBc$ is complete and separable in the bottleneck metric.
\end{thm}

\begin{proof}[Proof of Theorems~\ref{thm:aBN} and~\ref{thm:aBc}]
The multi-sets of intervals with rational endpoints provides a
countable dense subset for $\aB_{N}$.  To see that $\aB$ is dense in
$\aBc$, given $A$ in $\aBc$ and $\epsilon >0$, let $A_{\epsilon}$ be
the multi-subset of $A$ of those intervals of length $>\epsilon$.
Then by definition of $\aBc$, $A_{\epsilon}$ is in $\aB$, and by
definition of the bottleneck metric, using the matching coming from
the inclusion of $A_{\epsilon}$ in $A$, we have that 
\[
d_{\aB}(A,A_{\epsilon}) \leq \epsilon /2 <\epsilon.
\] 
It just remains to prove completeness of $\aB_{N}$ and $\aBc$.  For
this, given a Cauchy sequence $\langle X_{n}\rangle$ in $\aB$ it suffices to show
that $X_{n}$ converges to an element $X$ in $\aBc$ and that $X$ is in
$\aB_{N}$ if all the $X_{n}$ are in $\aB_{N}$.

Let $\langle X_{n}\rangle$ be a Cauchy sequence in $\aB$.  By passing
to a subsequence if necessary, we can assume without loss of
generality that for $n,m>k$, $d_{\aB}(X_{m},X_{n})<2^{-(k+2)}$.  For
each $n$, we have $d_{\aB}(X_{n},X_{n+1})<2^{-(n+1)}$; choose a matching
$C_{n}$ such that $d_{\infty}(I,J)<2^{-(n+1)}$ for all $(I,J)\in C_{n}$.
For each $n$, define a finite sequence of intervals $I^{n}_{1}$,\dots,
$I^{n}_{k_{n}}$ inductively as follows.  Let $k_{0}=0$.  Let $k_{1}$ be the
cardinality of the multi-subset of $X_{1}$ consisting of those
intervals of length $>1$, and let $I^{1}_{1}$,\dots,
$I^{1}_{k_{1}}$ be an enumeration of those intervals.  By induction,
$I^{n}_{1}$,\dots, $I^{n}_{k_{n}}$ is an enumeration of the intervals
in $X_{n}$ of length $>2^{-n+1}$ such that for $i\leq k_{n-1}$, the
intervals $I^{n-1}_{i}$ and $I^{n}_{i}$ correspond under the matching
$C_{n-1}$.  For the inductive step, we note that if $I^{n}_{i}$
corresponds to $J$ under $C_{n}$, then
$d_{\infty}(I^{n}_{i},J)<2^{-(n+1)}$, so the length $||J||$ of $J$ is
bigger than $||I^{n}_{i}||-2^{-n}$, and
\[
||J||>2^{-n+1}-2^{-n}=2^{-n}=2^{-(n+1)+1}.
\]
Thus, we can choose $I^{n+1}_{i}$ to be the corresponding interval $J$
for $i\leq k_{n}$, and we can choose the remaining intervals of length
$>2^{-(n+1)+1}$ in an arbitrary order.   Write
$I^{n}_{i}=[a^{n}_{i},b^{n}_{i})$ and let 
\[
a_{i}=\lim_{n\to \infty}a^{n}_{i}, \qquad b_{i}=\lim_{n\to\infty}b^{n}_{i}.
\]
Since $|a^{n}_{i}-a^{n+1}_{i}|<2^{-(n+1)}$ and
$|b^{n}_{i}-b^{n+1}_{i}|<2^{-(n+1)}$, we have   
\[
|a^{n}_{i}-a_{i}|\leq 2^{-n},\qquad |b^{n}_{i}-b_{i}|\leq 2^{-n}.
\]
Let $X$ be the multi-subset of $\aI$ consisting of the intervals
$I_{i}=[a_{i},b_{i})$ for all $i$ (or for all $i\leq \max k_{n}$ if
$\{k_{n}\}$ is bounded).

First, we claim that $X$ is in $\aBc$.  Given $\epsilon >0$, choose
$N$ large enough that $2^{-N+2}<\epsilon$.  Then for $i>k_{N}$, the
interval $I_{i}$ first appears in $X_{n_{i}}$ for some $n_{i}>N$.
Looking at the matchings $C_{N}$,\dots, $C_{n_{i}-1}$, we get a
composite matching $C_{N,n_{i}}$ between $X_{N}$ and $X_{n_{i}}$.
Since each $C_{n}$ satisfied the bound $2^{-(n+1)}$, the matching
$C_{N,n_{i}}$ must satisfy the bound 
\[
\sum_{n=N}^{n_{i}-1}2^{-(n+1)}=2^{-N}-2^{-n_{i}}.
\]
Since all intervals of length $>2^{-N+1}$ in $X_{N}$ appear as an
$I^{N}_{j}$, we must have that the length of $I^{n_{i}}_{i}$ in
$X_{n_{i}}$ must be less than 
\[
2^{-N+1}+2(2^{-N}-2^{-n_{i}})=2^{-N+2}-2^{-n_{i}+1}.
\]
Since each
endpoint in $I_{i}$ differs from the endpoint of $I^{n_{i}}_{i}$ by at
most $2^{-n_{i}}$, the length of $I_{i}$ can be at most
\[
2^{-N+2}-2^{-n_{i}+1}+2\cdot 2^{-n_{i}}=2^{-N+2}<\epsilon.
\]
Thus, the cardinality of the multi-subset of $X$ of those intervals of
length $>\epsilon$ is at most $k_{N}$.

Next we claim that $\langle X_{n}\rangle$ converges to $X$.  We have a
matching of $X_{n}$ with $X$ given by matching the intervals
$I^{n}_{1}$,\dots, $I^{n}_{k_{n}}$ in $X_{n}$ with the corresponding
intervals $I_{1}$,\dots, $I_{k_{n}}$ in $X$.  Our estimates above for
$|a^{n}_{i}-a_{i}|$ and $|b^{n}_{i}-b_{i}|$ show that
$d_{\infty}(I^{n}_{i},I_{i})\leq 2^{-n}$.  By construction, each
leftover interval in $X_{n}$ has length $\leq 2^{-n+1}$ and the
previous paragraph shows that each leftover interval in $X$ has length
$<2^{-n+2}$.  Thus, $d_{\aB}(X_{n},X)<2^{-n+1}$.

Finally we note that if each $X_{n}$ is in $\aB_{N}$ for fixed $N$,
then each $k_{n}\leq N$ and so $X$ is in $\aB_{N}$.
\end{proof}

\section{Failure of robustness}\label{sec:notrobust}

Inevitably physical measurements will result in bad samples.  As a
consequence, we are interested in invariants which have limited
sensitivity to a small proportion of arbitrarily bad samples.  Many
standard invariants not only have high sensitivity to a small
proportion of bad samples, but in fact have high sensitivity to a
small number of bad samples.  We do not claim particular novelty for
the general nature of the results of this section, as these issues have
been folklore for some time.  However, we do not know any place in the
literature where precise statements are written down.  We use the
following terminology to describe the instability of these invariants.

\begin{defn}\label{defn:unrobust}
A function $f$ from the set of finite metric spaces to $\bR$ is {\em
fragile} if it not robust (in the sense of
Definition~\ref{defn:robust}) for any robustness coefficient $r > 0$.
\end{defn}

In some cases, an even stronger kind of sensitivity holds.

\begin{defn}\label{defn:notrobust}
A function $f$ from the set of finite metric spaces to $\bR$ is {\em
extremely fragile} if there exists a constant $k$ such that for every
non-empty finite metric space $X$ and constant $N$ there exists a
metric space $X'$ and an isometry $X \to X'$ such that $|X'| \leq |X|
+ k$ and $|f(X') - f(X)| > N$. 
\end{defn}

Informally, extremely fragile in Definition~\ref{defn:notrobust} means
that adding a small constant number of points to any metric space can
arbitrarily change the value of the invariant.  In particular, an
extremely fragile function is fragile, but extremely fragile is much more
unstable than just failing to be robust (note the quantifier on the
space $X$).  As we indicated in the introduction, Gromov-Hausdorff
distance is fragile; here we show it is extremely fragile.

\begin{prop}
Let $(Z,d_Z)$ be a non-empty finite metric space.  The function $d_{GH}(Z,-)$
is extremely fragile.
\end{prop}

\begin{proof}
Given $N>0$, consider the space $X'$ which is defined as a set to be
the disjoint union of $X$ with a new point $w$, and made a metric space by setting 
\begin{align*}
&d(w, x) = \alpha, && x \in X,\\
&d(x_1, x_2) = d_X(x_1, x_2),&& x_1,x_2 \in X,
\end{align*}
where $\alpha >\diam(Z)+2d_{GH}(Z,X)+2N$.  We claim 
\[
|d_{GH}(Z,X)-d_{GH}(Z,X')|>N.
\]
Given any metric space $(Y,d_{Y})$
and isometries $f\colon X'\to Y$, $g\colon Z\to Y$, we need to show
that $d_{Y}(f(X'),g(Z))>N+d_{GH}(Z,X)$.  We have two cases.
First, if no point $z$ of $Z$ has $d_{Y}(g(z),f(w))\leq
N+d_{GH}(Z,X)$, then we have $d_{Y}(f(X'),g(Z))>N+d_{GH}(Z,X)$.
On the other hand, if some point $z$ of $Z$ has
$d_{Y}(g(z),f(w))<N+d_{GH}(Z,X)$, then every
point $z$ in $Z$ satisfies $d_{Y}(g(z),f(w))\leq N+\diam(Z)+d_{GH}(Z,X)$.  Choosing
some $x$ in $X$, we see that for every $z$ in $Z$, 
$d_{Y}(f(x),g(z))\geq \alpha-(N+\diam(Z)+d_{GH}(Z,X))$.  It follows that 
\[
d_{Y}(f(X'),g(Z))\geq \alpha-(N+\diam(Z)+d_{GH}(Z,X)) >
N+d_{GH}(Z,X). \qedhere
\]
\end{proof}

The homology and persistent homology of a point cloud turns out to be
a somewhat less sensitive invariant.  Nonetheless, a similar kind of
problem can occur.  It is instructive to consider the case of $H_0$ or
$\PH_0$.  By adding $\ell$ points far from the original metric space
$X$, one can change either $H_0$ or $\PH_0$ by rank $\ell$.  The
further the distance of the points, the longer the additional bars in
the barcode and we see for example that the distance $d_{\aB}(B,-)$ in the
bottleneck metric from any fixed barcode $B$ is a extremely fragile
function.  (If we are truncating the barcodes, $d_{\aB}$ is bounded by
the length of the interval we are considering, so technically is
robust, but not in a meaningful way.)  We can also consider the rank
of $H_{0}$ or of $\PH_{0}$ in a range; here the distortion of the
function depends on the number of points, but we see that the function
is fragile.

For $H_k$ and $\PH_k$, $k \geq 0$, the same basic idea obtains: we add
small spheres sufficiently far from the core of the points in order to
adjust the required homology.  We work this out explicitly for
$\PH_1$.

\begin{defn}
For each integer $k > 0$ and real $\ell > 0$, let the metric circle
$S^1_{k,\ell}$ denote the metric space with $k$ points $\{x_i\}$ such
that 
\[
d(x_i, x_j) = \ell \left( \min (|i-j|, |k-i-j|) \right).
\]
\end{defn}

For $\epsilon < \ell$, the Vietoris-Rips complex associated to $S^1_{k,\ell}$
is just a collection of disconnected points.  It is clear that as long
as $k \geq 4$, when $\ell \leq \epsilon < 2\ell$,
$|R_\epsilon(S^1_{k,\ell})|$ has the homotopy type of (and is in fact
homeomorphic to) a circle.  In fact, we can say something more
precise:

\begin{lem}
For 
\[
\ell \leq \epsilon < \left\lceil \frac{k}{3}\right\rceil \ell,
\]
the rank of $H_1(R_{\epsilon}(S^1_{k,\ell}))$ is at least $1$.
\end{lem}

\begin{proof}
Consider the map $f$ from $R_{\epsilon}(S^{1}_{k,\ell})$ to the unit
disk $D^{2}$ in $\bR^{2}$ that sends $x_{i}$ to
$(\cos(2\pi\frac{i}{n}),\sin(2\pi\frac{i}{n}))$ and is linear on each
simplex.  The condition $\epsilon < \lceil \frac{k}{3}\rceil \ell$
precisely ensures that whenever $\{x_{i_{1}},\dotsc,x_{i_{n}}\}$ forms
a simplex $\sigma$ in the Vietoris-Rips complex, the image vertices
$f(x_{i_{1}}),\dotsc,f(x_{i_{n}})$ lie on an arc of angle
$<\frac{2}{3}\pi$ on the unit circle, and so $f(\sigma)$ in particular
lies in an open half plane through the origin.  It follows that the
origin $(0,0)$ is not in the image of any simplex, and $f$ defines a
map from $R_{\epsilon}(S^{1}_{k,\ell})$ to the punctured disk
$D^{2}-\{(0.0)\}$.  Since $\ell \leq \epsilon$, we have the $1$-cycle
\[
[x_{1},x_{2}]+\dotsb +[x_{k-1},x_{k}]+[x_{k},x_{1}]
\]
of $R_{\epsilon}(S^{1}_{k,\ell})$ which maps to a $1$-cycle in $D^{2}-\{(0,0)\}$
representing the generator of $H_{1}(D^{2}-\{0,0\})$.
\end{proof}

The length $\ell$ and number $k\geq 4$ is arbitrary, so again, we
conclude that functions like $d_{\aB}(B,\PH_{1}(-))$ are extremely
fragile.  Results for higher dimensions (using similar standard
equidistributed models of $n$-spheres) are completely analogous.

\begin{prop}
Let $B$ be a barcode.  The functions $d_{\aB}(B,\PH_{k}(-))$ from finite
metric spaces to $\bR$ are extremely fragile.
\end{prop}

In terms of rank, the lemma shows that we can increase the rank of
first persistent homology group of a metric space $X$ on an interval
$[a,b]$ by $m$ simply by adding extra points.  One can also typically
reduce persistent homology intervals by adding points ``in the
center'' of the representing cycle.  It is somewhat more complicated
to precisely analyze the situation, so we give a representative
example: Suppose the cycle is represented by a collection of points
$\{x_i\}$ such that the maximum distance $d(x_i, x_j) \leq \delta$.
Then adding a point which is a distance $\delta$ from each of the
other points reduces the lifetime of that cycle to $\delta$.  In any
case, the results of the lemma are sufficient to prove the following
proposition.

\begin{prop}
The function that takes a finite metric space to the rank of $\PH_{k}$
on a fixed interval $[a,b]$ is fragile.
\end{prop}

These computations suggest a problem with the stability of the usual
invariants of computational topology.  A small number of bad samples
can lead to arbitrary changes in these invariants.

\section{The main definition and theorem}\label{sec:distinv}

Fix a good functorial assignment of a simplicial complex to a finite
metric space and a scale parameter $\epsilon$.  Recall that we write
$\PH_k$ of a finite metric space to denote the persistent homology of
the associated direct system of complexes.  Motivated by the concerns
of the preceding section, we define $\PHD_k^n$ as the
distribution of barcodes induced by samples of size $n$.  The basic
idea motivating $\PHD_k^n$ is that in order to obtain robust
invariants, given a sample budget of $N$ samples from
$(X, \partial_X, \mu_X)$, instead of computing a single estimator from
the $N$ samples it is preferable to look at {\em the distribution} of
estimators produced by blocks of samples of size $n \ll N$.  Note that
this is closely related to the idea behind bootstrap
resampling.  It is also a more sophisticated version of computing a
trimmed mean (i.e., a mean in which extremal samples are thrown out)
--- rather than removing extremal samples, we simply subsample at a
rate such that the likelihood of seeing a bad sample is low.  Ideally,
this approach retains the information contained in those samples while
also estimating the ``true'' value.

\begin{defn}\label{defn:PHD}
For a metric measure space $(X, \partial_X, \mu_X)$ and fixed $n,
k \in \bN$, define the $k$th $n$-sample persistent homology as
\[
\PHD_{k}^{n}(X, \partial_X, \mu_X) = (\PH_k)_* (\mu_X^{\otimes n}), 
\] 
the probability distribution on $\aBc$ induced by pushforward along
$\PH_k$ from the product measure $\mu_X^{n}$ on $X^n$.
\end{defn}

This definition makes sense because $\PH_k$ is a continuous function
and the measures on the domain and codomain are both Borel.  Indeed, the
stability theorem of Chazal, et. al.~\cite[3.1]{chazal}
(Theorem~\ref{thm:GHcompare} above) and the fact that the Gromov Hausdorff
metric is less than or equal to the product metric in $X^{n}$
implies that $\PH_k$ is Lipschitz with Lipschitz constant at most 1.

In order to apply $\PHD_k^n$, we need to know two things.  First, that
for fixed $n$ and $k$ the approximation to $\PHD_k^n$ computed by
choosing samples from the empirical measure on a large sample space
$S$ drawn from $(X, \partial_X, \mu_X)$ converges in probability to
the actual value (as $|S|$ goes to infinity).  Second, that for fixed
$n$ and $k$ the approximation to $\PHD_k^n$ obtained by computing the
empirical measure from $\ell$ blocks of $n$ samples converges in
probability to the actual value (as $\ell$ goes to infinity).  The
latter follows from the weak law of large numbers for the empirical
process.  The goal of this section is to prove the following theorem,
which establishes the former asymptotic consistency.  For this (and in
the remainder of the section), we assume that we are computing $\PH$
using the Vietoris-Rips complex.

\begin{thm}\label{thm:main}
Let $(X, \partial_X, \mu_X)$ and $(Y, \partial_Y, \mu_Y)$ be compact
metric measure spaces.  Then we have the following inequality:
\[
d_{Pr}(\PHD_{k}^{n}(X, \partial_X, \mu_X), \PHD_{k}^{n}(Y, \partial_Y, \mu_Y)) \leq n \,
d_{GPr}((X, \partial_X, \mu_X), (Y, \partial_Y, \mu_Y)).
\]
\end{thm}

\begin{proof}
Assume that $d_{GPr}((X, \partial_X, \mu_X), (Y, \partial_Y, \mu_Y))
< \epsilon$.  Then we know that there exist embeddings $\iota_1 \colon
X \to Z$ and $\iota_2 \colon Y \to Z$ into a metric space $Z$ and a
coupling $\hat{\mu}$ between $(\iota_1)_* \mu_X$ and
$(\iota_2)_* \mu_Y$ such that the probability mass of the set of 
pairs $(z,z')$ under $\hat{\mu}$ such that
$\partial_Z(z,z') \geq \epsilon$ is less than $\epsilon$.

We can regard the restriction of $\hat{\mu}^{\otimes n}$ to the full measure
subspace $(X\times Y)^{n}$ of $(Z\times Z)^{n}$ as a probability
measure on $X^n \times Y^n$.  This then induces a
coupling between $(\PH_k)_*(\mu_X^{\otimes n})$ and $(\PH_k)_*(\mu_Y^{\otimes
n})$ on $\aBc$, which we now study.  
Consider $n$ samples $\{(x_1, y_1), (x_2, y_2), \ldots, (x_n, y_n)\}$
from $Z\times Z$
drawn according to the product distribution $\hat{\mu}^{\otimes n}$.
Now consider the probability that
\[
\alpha = \sup_{1 \leq i,j \leq n} |\partial_X(x_i,x_j)
- \partial_Y(y_i, y_j)| \geq 2 \epsilon.
\]
The triangle inequality implies that
\[
|\partial_X(x_i,x_j) - \partial_Y(y_i, y_j)|
=
|\partial_Z(x_i,x_j) - \partial_Z(y_i, y_j)|
\leq \partial_Z(x_i,y_i)
 + \partial_Z(x_j, y_j).
\] 
Therefore, the union bound implies that the probability that
$\alpha \geq 2\epsilon$ is bounded by
\[
\Pr(\exists i\mid \partial_Z(x_i,y_i) \geq \epsilon)\leq
1-(1-\epsilon)^{n} < n\epsilon  
\]
Next, define a relation $R$ that matches $x_i$ and $y_i$.  By
definition, the distortion of this relation is $\dis R = \alpha$, and
so 
\[
d_{GH}(\{x_i\},\{y_i\}) \leq \frac{1}{2} \alpha.
\]
By the stability theorem of Chazal, et. al.~\cite[3.1]{chazal} (Theorem~\ref{thm:GHcompare} above), this implies that the probability that 
\[
d_{\aB}(\PH_{k}(\{x_i\}), \PH_{k}(\{y_i\})) \geq \epsilon
\]
is bounded by $n \epsilon$.  This further implies that the probability
that 
\[
d_{\aB}(\PH_{k}(\{x_i\}), \PH_{k}(\{y_i\})) \geq n\epsilon
\]
is also bounded by $n\epsilon$.  Therefore, we can conclude that
\[
d_{Pr}(\PHD_{k}^{n}(X, \partial_X, \mu_X), \PHD_{k}^{n}(Y, \partial_Y, \mu_Y)) \leq
n\epsilon.\qedhere
\]
\end{proof}

We note the dependence on $n$ in the statement of the bound in
Theorem~\ref{thm:main}.  As $n$ increases, the quantity $\PHD_k^n$
becomes a finer approximation of the persistent homology of the
support of $X$.  Specifically, more points per sample means that
$\PHD_k^n$ is increasingly sensitive to small features of $X$.  In
this light, it is not surprising that the bound becomes weaker for
larger $n$.

Next we discuss the tightness of the bound in Theorem~\ref{thm:main}.
Clearly, this bound is vacuous when
$d_{GPr}((X, \partial_X, \mu_X), (Y, \partial_Y, \mu_Y))
> \frac{1}{n}$ since the Prohorov metric is bounded by $1$, but we
show that it becomes tight as 
$d_{GPr}((X, \partial_X, \mu_X), (Y, \partial_Y, \mu_Y))$
approaches zero.  Reviewing the argument, starting from the hypothesis
that\break $d_{GPr}((X, \partial_X, \mu_X), (Y, \partial_Y, \mu_Y))
= \epsilon$, we used the union bound to obtain a bound of
$n\epsilon$.  The exact bound in question is $1 - (1 - \epsilon)^n$.
The leading term in the expansion of this quantity is $n\epsilon$, and
so as $\epsilon \to 0$ the bound in the theorem becomes increasingly
tight.  When $\epsilon$ is close to $\frac{1}{n}$, using more terms in the
expansion yields better bounds (for example, when $\epsilon =\frac1n$,
$1-(1-\epsilon)^{n}\leq .75$ and tends to $1-\frac1e\approx .632$ for large $n$).

The exact bound $1-(1-\epsilon)^{n}$ yields a tight estimate on
$d_{GPr}(X^{n},Y^{n})$ (using the $\sup$ product metric), as we 
can see by the following example.  Consider the case of two finite
metric spaces $X = X_1 \cup X_2$ and $Y = Y_1 \cup Y_2$, where $|Y_1|
= |X_1|$ and $|Y_2| = |X_2|$.  Define $d_X$ via $d_X(x_1, x_1') =
\alpha$ for $x_1, x_1' \in X_1$, $d_X(x_2, x_2') = \beta$ for $x_2,
x_2' \in X_2$, and $d_X(x_1, x_2) = \gamma$ for $x_1 \in X_1$ and $x_2
\in X_2$.  Here $\gamma$ should be substantially larger than $\alpha$
and $\beta$.  We define $d_Y$ analogously, using the same $\alpha$ and
$\beta$ but with $\gamma'$ distinct from $\gamma$ (and without loss of
generality assume that $\gamma' > \gamma$).  Consider the metric space
$Z$ formed from the disjoint union of $X_1$, $X_2$, and $Y_2$, and
with the metric induced from $d_X$ and $d_Y$ except that $d_Z(x_2,
y_2) = \gamma' - \gamma$.  There are evident isometries $i\colon X \to Z$ and
$j\colon Y \to Z$; it is easy to see that
$d_{Pr}(i_{*}\mu_{X},j_{*}\mu_{Y})=\epsilon $ for $\epsilon
=\frac{|X_2|}{|X_1|+|X_2|}$ and moreover that this pair of embeddings
minimizes the Prohorov distance, so
$d_{GPr}(X,Y)=\epsilon$.  The induced embeddings $i^{n}\colon X^{n}\to
Z^{n}$, $j^{n}\colon Y^{n}\to Z^{n}$ satisfy
\[
d_{Pr}(i^{n}_{*}\mu_{X}^{\otimes n},j^{n}_{*}\mu_{Y}^{\otimes n})
=1-(1-\epsilon)^{n}
\]
and a straight-forward combinatorics argument shows that this
embedding also minimizes the Prohorov distance, so
$d_{GPr}(X^{n},Y^{n})=1-(1-\epsilon)^{n}$.  (We thank Olena Blumberg
for help with this example.)

By varying the parameters in the previous example, it is now clear
that the bound on $\PHD^{n}_{0}$ is tight and we can achieve the upper
bound with a variety of barcode lengths.  Tightness for
$\PHD^{n}_{k}$ for $k>0$ is harder to analyze.
Theorem~\ref{thm:GHcompare} is expected to be tight for all $k$ but
nothing has yet appeared in the literature for $k>0$.  If the bound in
Theorem~\ref{thm:GHcompare} is tight, it is reasonable to expect the
bound in Theorem~\ref{thm:main} also to be tight; however, we do not
know a rigorous argument.

\begin{rem}
For a related discussion involving finite distance matrices,
see~\cite[\S6, \S7]{gadgil}.  There the constant $N$ (size of the
matrix) is analogous to the parameter $n$ above and enters into their
formulas through the distance $d_{M}$, which depends on $N$.
\end{rem}

We regard the bound as most useful for fixed $n$.  Then a basic
consequence of Theorem~\ref{thm:main} is that consideration of large
finite samples will suffice for computing $\PHD_k^n$.  For a finite
metric space $(X,\partial_X)$, let $\mu$ and $\mu'$ denote two
measures on $X$.  Then we have the following
inequality~\cite[5.4]{gadgil}
\begin{equation}\label{eq:ineq}
d_{GPr}((X, \partial_X, \mu), (X, \partial_X, \mu')) \leq 1
- \sum_{x \in X} \min(\mu(x), \mu'(x)),
\end{equation}
which follows by choosing a coupling that has measure at least the
minimum of the two measures on each point.

\begin{cor}\label{cor:largenum}
Let $S_1 \subset S_2 \subset \cdots \subset S_i \subset \cdots$ be a
sequence of randomly drawn samples from $(X, \partial_X, \mu_X)$.  We
regard $S_i$ as a metric measure space using the subspace metric and
the empirical measure.  Then $\PHD_{k}^{n}(S_i)$ converges in
probability to $\PHD_{k}^{n}(X, \partial_X, \mu_X)$.
\end{cor}

\begin{proof}
This result is a consequence of the fact that $\{S_i\}$ converges in
probability to $(X, \partial_X, \mu_X)$ in the Gromov-Prohorov metric
(which can be checked directly using equation~\eqref{eq:ineq}, as
in~\cite[\S 5]{gadgil}, or can be deduced from the analogous
convergence result for the Gromov-Wasserstein
distance~\cite[3.5.(iii)]{sturm} and the comparison between the
Gromov-Prohorov distance and the Gromov-Wasserstein
distance~\cite[10.5]{greven}).
\end{proof}

Another consequence of Theorem~\ref{thm:main} is that
$\PHD_{k}^{n}$ provide robust descriptors for metric measure spaces
$(X, \partial_X, \mu_X)$.  Specifically, observe that if we have
finite metric spaces $(X,\partial_{X})\subset (X',\partial_{X'})$ and
a probability measure $\mu_{X'}$ on $X'$ that restricts to $\mu_{X}$
on $X$ (i.e., for $A\subset X$, $\mu_{X}(A)=\mu_{X'}(A)/\mu_{X'}(X)$),
then equation~\eqref{eq:ineq} implies that 
\[
d_{Pr}(i_{*}\mu_{X},\mu_{X'})\leq 1-\mu_{X'}(X).
\]
Thus, when $X'\setminus X$ has probability $<\epsilon$,
\[
d_{Pr}(\PHD_{k}^{n}(X, \partial_X, \mu_X), \PHD_{k}^{n}(X', \partial_X^{'}, \mu_{X'})) \leq
n\epsilon.
\]
In particular, when $X$ and $X'$ are finite metric spaces with the
uniform measure, we get 
\[
d_{Pr}(\PHD_{k}^{n}(X, \partial_X, \mu_X), \PHD_{k}^{n}(X', \partial_X^{'}, \mu_{X'})) \leq
n(1-|X|/|X'|).
\]
As an immediate consequence we obtain the following result.

\begin{thm}\label{thm:mainunifrob}
For fixed $n,k$, $\PHD_{k}^{n}$ is uniformly robust with robustness
coefficient $r$ and estimate bound $nr/(1+r)$ for any $r$.  
\end{thm}

\begin{rem}
It would be useful to prove analogues of the main theorem for other
methods of assigning complexes; e.g., the witness complex (see
Remark~\ref{rem:witness} and Section~\ref{sec:notklein}).  We expect
that the recent stability results of~\cite{witstab} will be useful in
this connection.  
\end{rem}

\section{Hypothesis testing, confidence intervals, and numerical
invariants}\label{sec:intervals}

In this section, we describe various ways to use $\PHD_k^n$ to perform
statistical inference about the homological invariants of a point
cloud.  The basic goal is to provide quantitative ways of saying what
observed barcodes or empirical barcode distributions ``mean''.  We are
predominantly interested in addressing two kinds of questions:

\begin{enumerate}
\item Are two given empirical barcode distributions coming from the
same underlying distribution?

\item Is a particular empirical barcode distribution consistent with the
hypothesis that the underlying distribution has $k$ ``long bars''?
\end{enumerate}

We approach both of these questions from the perspective of classical
hypothesis testing, likelihood scores, and confidence intervals; for
example, see~\cite[\S 2]{conover} for a review.  We discuss a variety
of test statistics derived from $\PHD_k^n$; thus, the use of these
procedures are justified in practice by Theorem~\ref{thm:main} (and
specifically Corollary~\ref{cor:largenum}).  Moreover, we are able to
use Theorem~\ref{thm:main} to show that many of the test statistics we
describe are robust.

We begin by explaining the basic procedure for computing
approximations to $\PHD_k^n$.  Corollary~\ref{cor:largenum} justifies
the treatment of $\PHD_k^n$ of the empirical measure on a sufficiently
large sample $S \subset X$ of size $N$ as a good approximation for
$\PHD_k^n(X, \partial_X, \mu_X)$.  (Note that the dependence on $n$ in
the bound in Theorem~\ref{thm:main} implies that we will have to
choose $n \ll N$ in order to expect reasonable results; see the
discussion in the next section for some examples of how to choose $n$.)
Next, we can estimate $\PHD_k^n$ on $S$ empirically via Monte Carlo
simulation, i.e., simply sampling blocks of $n$ samples from $S$ over
and over again.  The weak law of large numbers for empirical
distributions guarantees that this estimate converges in probability
as the number of such samples increases.  Therefore, we have an
asymptotically convergent scheme for numerically approximating
$\PHD_k^n$ (and hence quantities derived from it).  We now turn to
questions of statistical inference.

\subsection{Hypothesis testing using $\PHD_k^n$}\label{sec:like}

The most basic question we can pose is whether a
given observed barcode $B$ is more consistent with an underlying
metric measure space $(X,\partial_X, \mu_X)$ or an alternate metric
measure space $(X', \partial_X', \mu_X')$.  A likelihood ratio
provides a good test statistic to determine an answer to this question.
Specifically, we can evaluate the likelihood of the hypotheses
$\Hyp^{n}_{k}(X;B, \epsilon)$ and $\Hyp^{n}_{k}(X';B, \epsilon)$ that
$B$ is within $\epsilon$ of a barcode drawn from
$(X,\partial_X, \mu_X)$ and $(X', \partial_{X'}, \mu_{X'})$
respectively.

Given an observed barcode $B$ (e.g., obtained by sampling $n$ points
from an unknown metric measure space $(Z, \partial_Z, \mu_Z)$), we can
compute the likelihood score
\[
L_Y = L(Y, \partial_Y, \mu_Y) = \Pr(d_{\aB}(B,\tilde{B}) < \epsilon
\mid \tilde{B} \text{\ drawn from\ } \PHD_{k}^{n}(Y, \partial_Y, \mu_Y))
\]
for each hypothesis space $(X, \partial_X, \mu_X)$ and
$(X', \partial_{X'}, \mu_{X'})$.  The ratio $L_X / L_{X'}$ then
provides a test statistic for comparing the two hypotheses.  To
determine how to interpret the test statistics (e.g., to compute
$p$-values), we require knowledge of the distribution of the test
statistic induced by assuming that $B$ was drawn from
$\PHD_{k}^{n}(X, \partial_X, \mu_X)$ and
$\PHD_{k}^{n}(X', \partial_{X'}, \mu_{X'})$ respectively.  These 
distributions can be approximated by Monte Carlo simulation, i.e.,
repeated sampling from the two distributions and computation of
histograms.

More generally, for a metric measure space $(X,\partial_{X},\mu_{X})$
and a particular subset $S$ of $\aBc$, we can test the hypothesis that
the distribution $\PHD_{k}^{n}(X)$ has mass $\geq \epsilon$ on $S$ as
follows.  For any hypothetical distribution on $\aBc$ with mass
$\geq \epsilon$ on $S$, the probability of an empirical sample of size
$N$ having $q$ or fewer elements in $S$ is bounded above by the
binomial cumulative distribution function
\[
\text{BD}(N,q,\epsilon)=\sum_{i=0}^{q}\binom{N}{i}\epsilon^{i}(1-\epsilon)^{N-i}.
\]
Then given an empirical approximation $\aE$ to $\PHD_{k}^{n}$ obtained
from $N$ samples, we can test the hypothesis that $\PHD_{k}^{n}$ has
mass $\geq \epsilon$ in $S$, by taking $q$ to be the number of such
elements in $\aE$.  When $\text{BD}(N,q,\epsilon)<\alpha$, we can
reject this hypothesis at the $1-\alpha$ level.

\subsection{Distribution comparison test statistics}

Another kind of question we might ask is to determine whether to
reject the hypothesis that two empirical distributions on
barcode space (i.e., $\PHD_k^n$ computed based on two different large
samples $S$ and $S'$) came from the same underlying distribution.  In
our setting we cannot assume very much about the class of possible
distributions and so we are forced to rely on non-parametric methods.
This imposes significant constraints --- most asymptotic results on
non-parametric tests for distribution comparison work only for
distributions on $\bR$.  Thus, the first step is to project the data
from barcode space into $\bR$.  The following definition is the first
of several kinds of projections we discuss.

\begin{defn}\label{defn:distdistr}
Let $(X, \partial_X, \mu_X)$ be a compact metric measure space.  Fix
$k,n \in \bN$.  
\begin{enumerate}
\item Define the distance distribution $\aD^{2}$ on $\bR$ to be the
distribution on $\bR$ induced by applying $d_{\aB}(-,-)$ to pairs $(b_1,
b_2)$ drawn from $\PHD_{k}^{n}(X,\partial_X, \mu_X)^{\otimes 2}$.  
\item Let $B$ be a fixed barcode in $\aBc$, and define $\aD_{B}$ to be
the distribution induced by applying $d_{\aB}(B,-)$.
\end{enumerate}
\end{defn}
Since both $\aD^2$ and $\aD_B$ are continuous with respect to the
Gromov-Prohorov metric~\cite[6.6]{greven},
Corollary~\ref{cor:largenum} justifies working with empirical
approximations to $\aD^{2}$ and $\aD_{B}$.

One application of these projections is simply a direct use of the
two-sample Kolmogorov-Smirnov statistic~\cite[\S6]{conover}.  This
test statistic gives a way to determine whether two observed empirical
distributions were obtained from the same underlying distribution; the
salient feature about this statistic is that for distributions on
$\bR$ the $p$-values of the test statistic are asymptotically
independent of the underlying distribution as long as the samples are
identically independently drawn.

To compute the Kolmogorov-Smirnov test statistic for two
sets of samples $\aS_{1}$ and $\aS_{2}$, we first compute the empirical
approximations $\aE_1$ and $\aE_2$ to the cumulative density
functions, 
\[
E_{i}(t)=|\{x\in S_{i}\mid x\leq t\}|/|S_{i}|,
\]
and use the test statistic $\sup_t |\aE_1(t) - \aE_2(t)|$.  In
practice, since $|S_{i}|$ is large we approximate $\aE_{i}$ using
Monte Carlo simulation.  The distribution-independence of the
statistic now implies that standard tables (e.g., in the appendix
to~\cite{conover}) or the built-in Matlab functions can then be used
to compute $p$-values for deciding if the statistic allows us to
reject the hypothesis that the distributions are the same.

One might similary consider the Mann-Whitney test or various other
nonparametric techniques for testing the same hypotheses~\cite[\S
5]{conover}.  For example, another way to handle this problem is to
use a $\chi^{2}$ test for discrete distributions.  There are many ways
to construct suitable distributions for this test; we present two
natural choices here.

\begin{enumerate}

\item Take histograms from $\aD^2_{S_1}$ and $\aD^2_{S_2}$ with
identical fixed numbers of bins and bin widths.

\item Fix a finite set $\{B_{j}\} \subset \aBc$ of reference barcodes,
where $1 \leq j \leq m$.  These reference barcodes should be chosen
without reference to the observed data.  Next, for each barcode with
nonzero probability measure in (the given empirical approximation to)
$\PHD_{k}^{n}$, assign the count to the nearest reference barcode.

\end{enumerate}

The second method makes sense if we have a priori information about
the expected shape of the barcode distributions.

Let $\aA_{i}(j)$ denote the count either for bin $j$ or for reference
barcode $B_{j}$ in sample $i$ (for $i= 1,2$).  The test statistic in
the $\chi^2$ test for comparing $S_1$ and $S_2$ is then defined to be
\[
\chi^{2} = \sum_{j=1}^{m} \frac{(\aA_{1}(j) -
\aA_{2}(j))^{2}}{\aA_{1}(j) + \aA_{2}(j)}.
\]
As the notation suggest, asymptotically this has a $\chi^{2}$
distribution with $m' - 1$ degrees of freedom (where $m'$ is the
number of reference barcodes with nonzero counts)~\cite[\S
17]{asympto}.  As such, we can again look up the $p$-values for this
distribution in standard tables when performing hypothesis testing.

\subsection{Numerical summaries as test statistics}

Natural test statistics for studying hypotheses about empirical
barcode distributions come from numerical summaries associated to
$\PHD_k^n$.  For instance, a natural test statistic measures the
distance to a fixed hypothesis distribution.

\begin{defn}\label{defn:HDD}
Let $(X, \partial_X, \mu_X)$ be a compact metric measure space and let
$\aP$ be a fixed reference distribution on $\aBc$.   Fix
$k,n \in \bN$.  Define the homological distance on $X$ relative to
$\aP$ to be
\[
\HD_k^n((X,\partial_X, \mu_X), \aP) =
d_{Pr}(\PHD_{k}^{n}(X,\partial_X, \mu_X), \aP).  
\]
\end{defn}

Corollary~\ref{cor:largenum} again applies to show that large
finite samples $S \subset X$ suffice to approximate $\HD_k^n$.  In
fact, the convergence is better since we are working over $\bR$ and
the Glivenko-Cantelli theorem applies.

\begin{lem}\label{lem:largenumHD}
Let $S_1 \subset S_2 \subset \cdots \subset S_i \subset \cdots$ be a
sequence of randomly drawn samples from $(X, \partial_X, \mu_X)$.
We regard $S_i$ as a metric measure space using the subspace metric
and the empirical measure.  Then for $\aP$ a fixed reference
distribution on $\aBc$, $\HD_k^n(S_{i}, \aP)$ converges almost surely
to $\HD_k^n((X,\partial_X, \mu_X), \aP)$.
\end{lem}

An immediate consequence of Theorem~\ref{thm:main} is the following
robustness result (paralleling Theorem~\ref{thm:mainunifrob}).

\begin{thm}\label{thm:mainunifrob2}
For fixed $n,k,\aP$, $\HD^{n}_{k}(-,\aP)$ is uniformly robust with
robustness coefficient $r$ and estimate bound $nr/(1+r)$ for any $r$.
\end{thm}

Another source of tractable test statistics is the moments of the
distributions introduced in Definition~\ref{defn:distdistr}.  A virtue
of distributions on $\bR$ is that they can be naturally summarized by
moments; in contrast, moments for distributions on barcode space are
hard to compute (for instance, see~\cite{frechet}).  Even simply
constructing meaningful centroids for a set of points in barcode space
is challenging; for example, geodesics between close points are not
unique, although the barcode metric space is a length space (it is
straightforward to construct midpoints between any pair of barcodes).
Because we have emphasized robust statistics, we work with the median
or a trimmed mean and introduce the following test statistics:

\begin{defn}\label{defn:ahdbody}
Let $(X, \partial_X, \mu_X)$ be a compact metric measure space and fix a reference barcode $B \in \aB$.   Fix $k,n \in \bN$.
Define the median homological distance relative to $B$ to be 
\[
\MHD_k^n ((X,\partial_X,\mu_X), B)
= \median(\aD^2).
\]
For $0 < \alpha < \frac{1}{2}$, define the $\alpha$-trimmed mean
homological distance to be  
\[
\widetilde{\MHD}\mhdstrut_k^n ((X, \partial_X, \mu_X), B) = \frac{1}{1 -
2\alpha} \int_{\alpha}^{1-\alpha} q(\aD^2),
\]
where $q$ denotes the quantile function.  (Roughly speaking, we
discard the fraction $\alpha$ of the highest and lowest values and
take the mean of the remainder.)
\end{defn}

Again, Corollary~\ref{cor:largenum} implies that consideration of
large finite samples $S \subset X$ suffices to approximate these test
statistics.

\begin{lem}\label{lem:largenumAHD}
Let $S_1 \subset S_2 \subset \cdots \subset S_i \subset \cdots$ be a
sequence of randomly drawn samples from $(X, \partial_X, \mu_X)$.  We
regard $S_i$ as a metric measure space using the subspace metric and
the empirical measure.  Let $B \in \aB$ be a fixed reference barcode.
\begin{enumerate}
\item Assume that $\aD^2 ((X,\partial_X, \mu_X), B)$ has a distribution
function with a positive derivative at the median.  Then
$\MHD_k^n(S_{i}, B)$ almost surely converges to $\MHD_k^n
((X,\partial_X, \mu_X), B)$.
\item $\widetilde{\MHD}\mhdstrut_k^n(S_{i}, B)$ almost surely converges to
$\widetilde{\MHD}\mhdstrut_k^n((X, \partial_X, \mu_X), B)$.
\end{enumerate}
\end{lem}

\begin{proof}
As in the proof of Corollary~\ref{cor:largenum}, the fact that
$\{S_i\}$ converges to $(X, \partial_X, \mu_X)$ in the Gromov-Prohorov
metric implies that $\aD(S_i)$ weakly converges to $\aD$.  Now the
central limit theorem for the sample median (see for
instance \cite[III.4.24]{pollard}) and the hypothesis about the
derivative at the median implies the convergence of medians.
Analogously, the central limit theorem for the trimmed
mean~\cite[\S4]{trimmed} (which holds without further assumption
provided that $\alpha < \frac{1}{2}$) gives the second part of the
result.
\end{proof}

The hypothesis on the median is the standard hypothesis for
consistency of the central limit theorem (and the bootstrap estimator
for) the sample median; it is known that this hypothesis is
required~\cite[5.11]{asympto}.  Although it is our experience that
this hypothesis holds in practice, it can be difficult to rigorously
verify for an unknown underlying distribution.  For this reason, the
use of the trimmed mean may be preferable in cases where constraint on
the possible hypotheses is unavailable.  As $\alpha$ approaches
$\frac{1}{2}$, the trimmed mean converges to the median, and so
choosing $\alpha = \frac{1}{2} - \epsilon$ for small $\epsilon$ yields
a reasonable alternative to the median.

As discussed in the introduction, a counting argument yields the
following robustness result.

\begin{thm}\label{thm:AHDrobustbody}
For any $n,k,B$, the function $\MHD_k^n(-,B)$ from finite metric
spaces (with the uniform probability measure) to $\bR$ is robust with
robustness coefficient $> (\ln 2)/n$.  
\end{thm}

\begin{proof}
For a finite metric space $X$, expanding $X$ to $X'$, the proportion
of $n$-element samples of $X'$ which are samples of $X$ is
$(|X|/|X'|)^{n}$; when this number is more than $1/2$, the median
value of any function $f$ on the set of $n$-element samples of $X'$ is
then bounded by the values of $f$ on $n$-element samples of $X$.
Since $(N/(N+rN))^{n}>1/2$ for $r<2^{1/n}-1$, any such function $f$
will be robust with robustness coefficient $r$ satisfying this bound,
and in particular for $r=(\ln 2)/n$.
\end{proof}

In order to obtain the $p$-value cutoffs for performing hypothesis
testing, we can again use Monte Carlo simulation to estimate the
distribution of these estimators under different hypotheses.  Another
possibility is to use asymptotic estimates, which we discuss below in
the context of confidence intervals.

We close the section by remarking that there are many other possible
numerical invariants one might associate to $\PHD_k^n$ (and apply as
test statistics).  For instance, we define for a barcode $B$ the
quantity 
\[
g_m(B) = |B(m)| - |B(m+1)|,
\]
where $B(i)$ denotes the $i$th largest interval in $B$.  Then the
quantity 
\[
g_m = \median(g_m(\PHD_k^n(-)))
\]
and the related quantity
\[
g = \max(g_m)
\]
are useful test statistics for determining if there is a group of
``long bars'' in the underlying distribution by checking for which (if
any $m$) has a large value of $g_{m}$.  For instance, when $g$
is small this suggests that the underlying metric measure space is
generated by a topological space with no homology in dimension $k$.
And large values of $g_m$ suggest that the underlying space has
rank $m$ homology in dimension $k$.  Of course, in order to make
precise statistical statements to replace ``suggests'', we have to use
Monte Carlo simulation in order to compute $p$-values and confidence
intervals.

\subsection{Confidence intervals}

For $\HD_k^n$, the only way to produce $p$-values and confidence
intervals is to use Monte Carlo simulation (to estimate the
distribution of $\HD_k^n$ on finite approximations to $\PHD^n_k$).
A particular advantage of $\MHD_k^n$ is that we can
define confidence intervals using the standard non-parametric
techniques for determining confidence intervals for the median and
trimmed mean~\cite[\S7.1]{david}.  For the median, we use appropriate
sample quantiles (order statistics) to determine the bounds for an
interval which contains the actual median with confidence $1
- \alpha$.  These confidence intervals then immediately yield cutoffs
for $p$-values for hypothesis testing.  For example, a simple
approximation can be obtained from the fact that order statistics
asymptotically obey binomial distributions, which lead to the
following definition using the normal approximation to the binomial
distribution.


\begin{defn}\label{defn:confint}
Let $(X,\partial_X, \mu_X)$ be a metric measure space and $B$ a fixed
barcode.  Fix
$0 \leq \alpha \leq 1$ and $n,k$.  Given $m$ samples from $\aD_{B}$, let
$\{s_m\}$ denote the samples sorted from smallest to largest.  Let
$u_{\alpha}$ denote the $\frac{\alpha}{2}$ significance threshold for
a standard normal distribution.  The $1- \alpha$ confidence interval
for the sample median (i.e., $\MHD_k^n$) is given by the interval
\[
\left[s_{\lfloor \frac{m+1}{2} - \frac{1}{2} \sqrt{m}
u_{\alpha}\rfloor}, s_{\lceil \frac{m+1}{2} + \frac{1}{2} \sqrt{m}
u_{\alpha}\rceil}\right].
\]
\end{defn}

For the trimmed mean, the situation is similar: asymptotic confidence
intervals can be obtained from the sample standard deviation and an
explicit formula~\cite{trimnonpar}.  Since we find that the median
converges in practice, we do not write out the formula here (as it
involves a number of complicated auxiliary quantities) and refer the
interested reader to the cited reference.

\subsection{The validity of asymptotic $p$-values}

In the preceding discussion, the $p$-values and confidence intervals
for our tests are always computed either via Monte Carlo simulation
(i.e., sampling to estimate the distribution of the test statistic) or
using formulas derived from asymptotic results.  The latter are
substantially easier and less computationally intensive to apply.
However, we may be concerned about whether sample sizes are large
enough for the asymptotic $p$-value to be good approximations of the
exact $p$-value --- this issue is a pervasive problem when applying
such non-parametric tests based on asymptotic results (e.g.,
see~\cite[\S 1.3]{asympto}).  

A standard approach to mitigating such concerns is to perform Monte
Carlo simulation of the distribution of these test statistics computed
from representative models for $\PHD_k^n$ (e.g., synthetic
distributions generated by various standard manifolds); this provides
heuristic guidance about suitable sample sizes.  We provide some
example calculations of this form in Section~\ref{sec:examples} below.
However, the careful analyst with access to adequate samples and
computer resources may simply choose to rely on Monte Carlo simulation
methods.  When adequate samples are lacking, resampling methods also
often provide a more reliable means to compute cutoffs than asymptotic
results.  For instance, standard results about the consistency of the
bootstrap for the sample median and sample trimmed mean~\cite{trimmed}
allow us to compute $p$-value thresholds and cofidence intervals for
$\MHD_k^n$ via bootstrap resampling.  

\section{Demonstration of hypothesis testing on synthetic
examples}\label{sec:examples}

In this section, we provide numerical experiments on synthetic data
sets to demonstrate the statistical inference procedures
and robustness results described in the previous section.  We study a
pair of examples embedded in $\mathbb{R}^2$ (an annulus and a pair of
nested circles) and three families of examples in $\mathbb{R}^3$
(spheres, tori, and uniform noise in a box).  Although the examples
embedded in $\mathbb{R}^2$ are essentially trivial, the simplicity of
the expected results allows us to focus on the methodology.  The
examples in $\mathbb{R}^3$ are more realistic but correspondingly are
more complicated to interpret.

All of our experiments rely on the following procedures for producing
empirical approximations to $\PHD_k^n$.  We fix a Monte Carlo
parameter $K$ which is large (we discuss estimates of how large $K$
needs to be below).  We then have the following basic algorithm:

\begin{alg}\label{alg:absolute}
For a fixed metric measure space $(X, \partial_X, \mu_X)$.
\begin{enumerate}
\item Uniformly select $K$ subsamples of size $n$ from $\mu_X$.
\item Compute the empirical approximation to $\PHD_k^n$ from the
$K$ subsamples.
\end{enumerate}
\end{alg}

To better represent the use of these procedures in practice, we have
the following variant algorithm.  Fix a subsample size $N$.

\begin{alg}\label{alg:relative}
For a fixed metric measure space $(X, \partial_X, \mu_X)$.
\begin{enumerate}
\item Uniformly sample $N$ points from $\mu_X$.
\item Uniformly select $K$ subsamples of size $n$ from the empirical
measure on the $N$ samples.
\item Compute the empirical approximation to $\PHD_k^n$ from the
$K$ subsamples.
\end{enumerate}
\end{alg}

To actually carry out these algorithms, we used the Perseus
codebase~\cite{perseus} to compute persistent homology from a finite
metric space, executed from within a series of Python and Cython
scripts that ran our various experimental setups.  In order to avoid
combinatorial explosion in the number of simplices when the scale
parameter results in complete graphs, we typically capped the maximum
scale parameter (i.e., truncated each of the bars in the barcodes).
The experiments were run on various stock Linux machines; no
individual experiment took more than a few minutes to complete.  Our
random number generation was done using the GSL library~\cite{GSL} to
produce uniform and Gaussian samples, and rejection sampling to
simulate all other distributions (as described below).

\subsection*{Synthetic Example 1: The annulus and the annulus plus
diameter linkage}

We first consider a simple example which illustrates the robustness of
the distributional invariants.  The underlying metric measure space
$A$ is an annulus of inner radius $0.8$ and outer radius of $1.2$ in
$\bR^2$ (see Figure~\ref{fig:ann}), equipped with the subspace metric
and the area measure.  The underlying manifold of $A$ is clearly
homotopy equivalent to a circle.

We sample from the annulus via rejection sampling; we sample uniformly
from the bounding box $[-1.2,1.2] \times [-1.2, 1.2]$ and only keep
points $(x,y)$ such that $0.8 \leq \sqrt{x^2 + y^2} \leq 1.2$.

We began by examining the rate of convergence in
Corollary~\ref{cor:largenum}.  Specifically, for $k = 1$ and various
$n$, we consider subsamples $S_i$ in the annulus of increasing
cardinality and study the convergence of various distributions derived
from $\PHD_k^n(S_i)$, using Algorithms~\ref{alg:absolute}
and~\ref{alg:relative} as a base.  We compute the distance
distribution $\aD^2$ from $\PHD_k^n(S_i)$ and $\PHD_k^n(A)$ as the
cardinality $N_i$ of $S_i$ increases and $n$ varies, using a barcode
cutoff of $0.375$.  We then used both the Kolmogorov-Smirnov test and
the $\chi^2$ test on histograms to repeatedly compare the estimates
computed from samples of cardinality $N_i$ to each other and to the
reference distribution from $A$.  Fixing $K = 1000$, our results
indicated that $|S_i| = 1000$ sufficed to approximate the distribution
for $n \le 100$; with these parameters, we were essentially never able
to reject the null hypothesis that the two distributions were drawn
from the same underlying distribution.

Next, we turn to an illustrative example of the behavior of
$\PHD_k^n$ in the face of maliciously chosen noise points.  We
generated sets $\aS_1$ of 1000 points by sampling uniformly (via
rejection sampling) from the annulus.  Using the Vietoris-Rips complex,
computing the barcode for the first homology group (with cutoff of
$0.375$) yields a single long interval, displayed in Figure~\ref{fig1}.
(We repeated this procedure many times with different subsamples of
size $1000$; the picture displayed is wholly representative of the
results, which varied only very slightly across the samples.)

\begin{figure}[htp]
\begin{center}
\hfill
\includegraphics{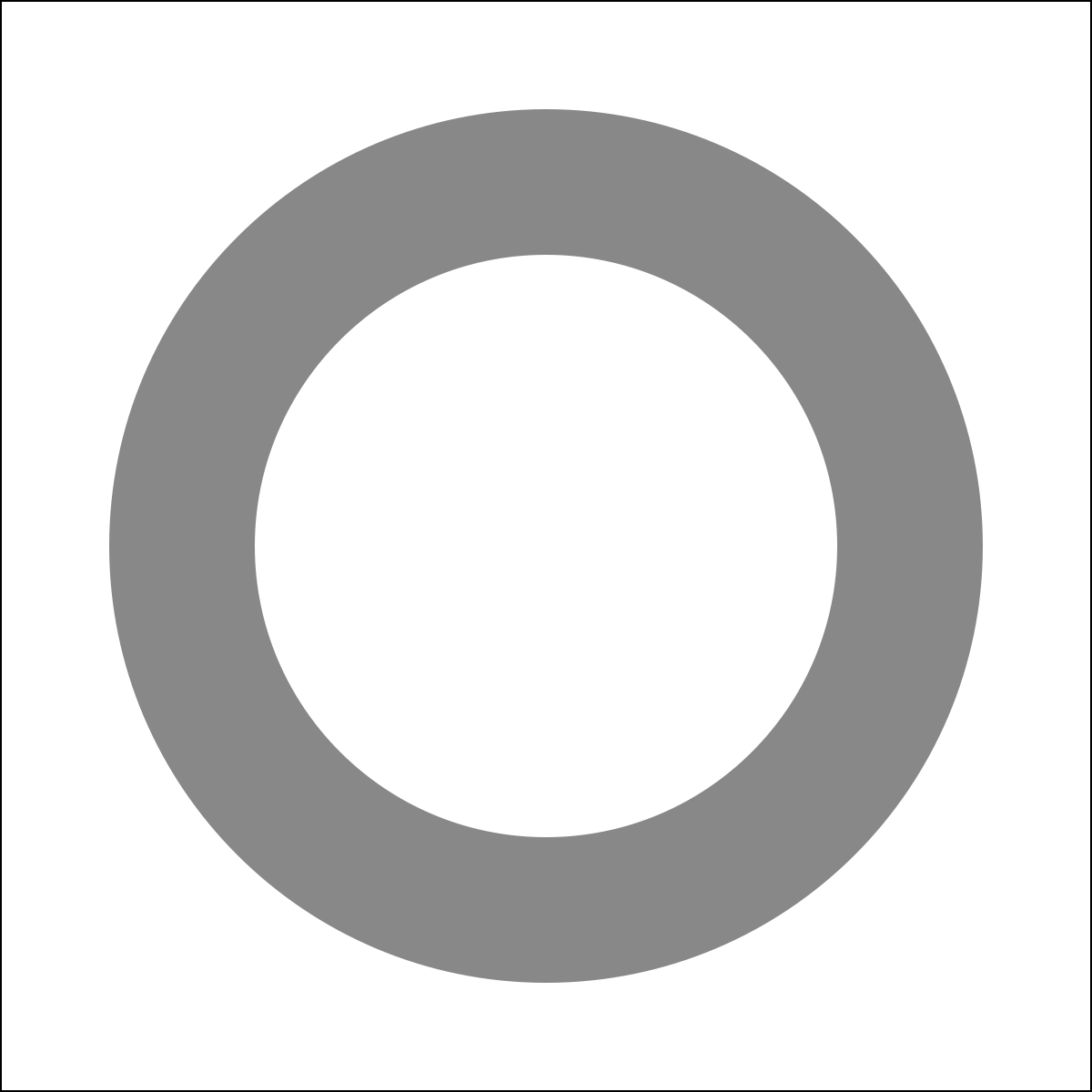}%
\hfill
\includegraphics{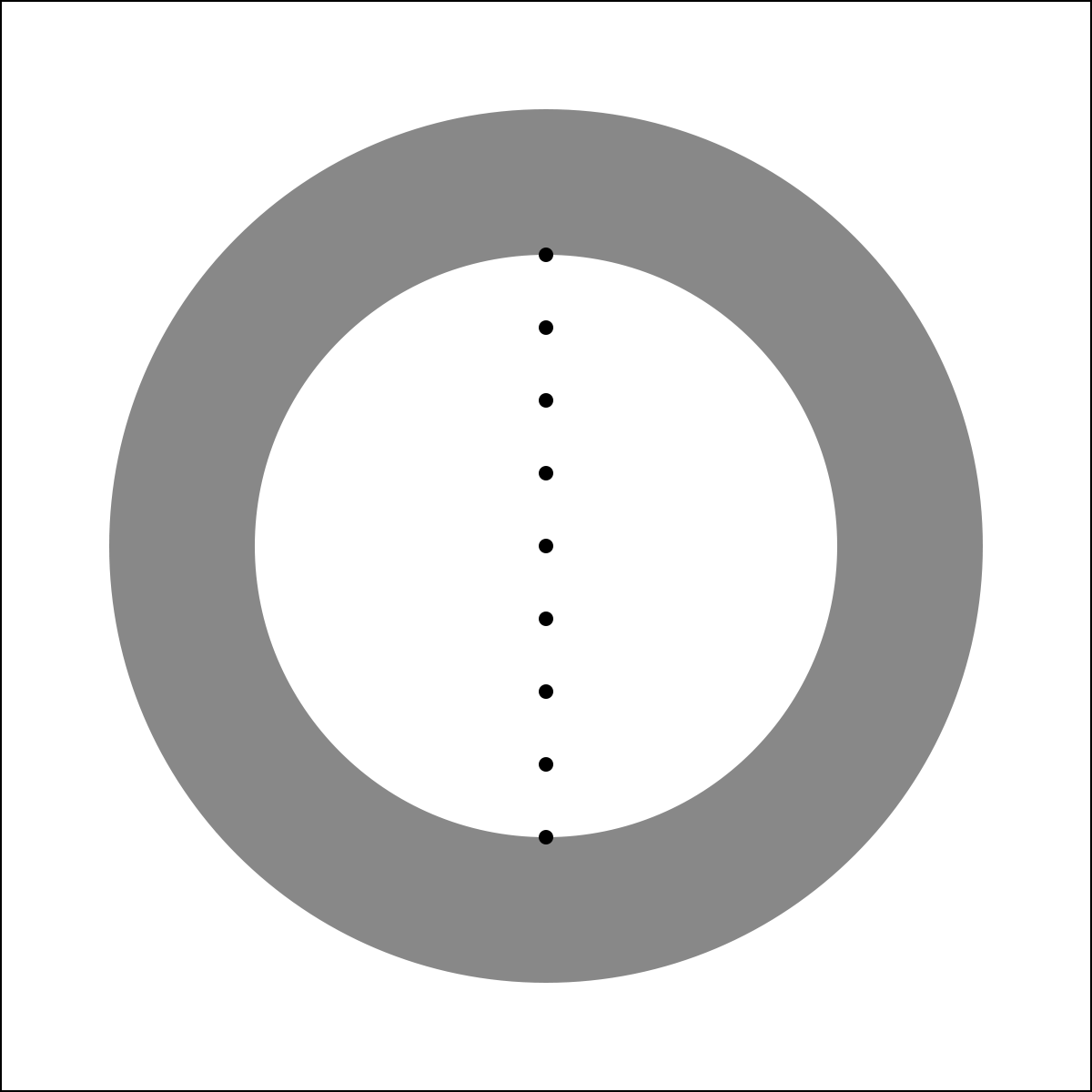}%
\hfill
\end{center}

\caption{Annulus with inner radius $0.8$ and outer radius $1.2$
(left), and same annulus together with diameter linkage (right). (In
experiments, points on diameter are chosen randomly.)}
\label{fig:ann}
\end{figure}

\begin{figure}[htp]
\begin{center}
\includegraphics{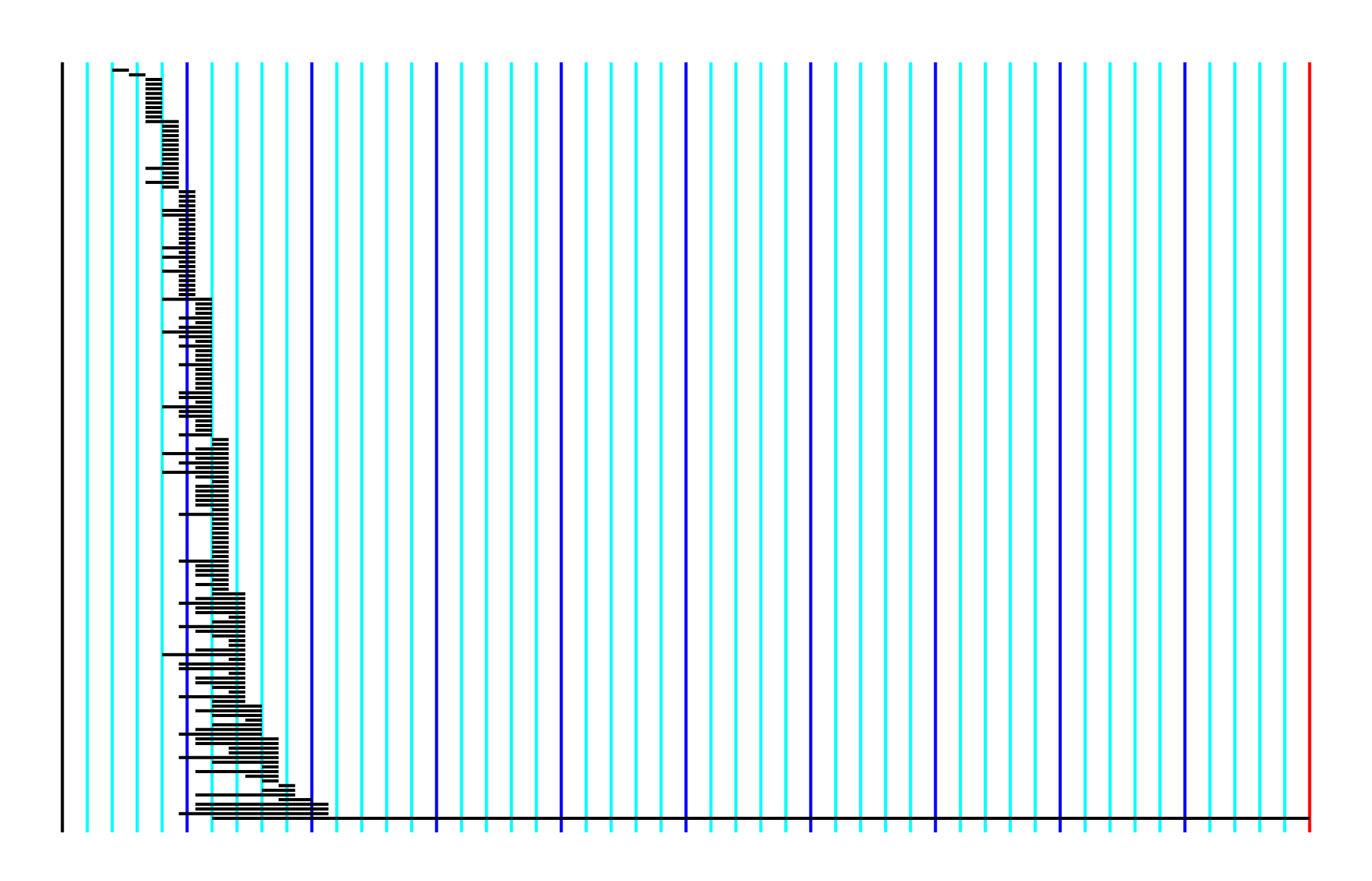}
\end{center}
\caption{Barcode for annulus via the Vietoris-Rips complex with 1000
points, showing 1 long bar.
Horizontal scale goes from 0 to 0.375. (Vertical scale is not meaningful.)}
\label{fig1}
\end{figure}

We then generated sets by drawing $\aS_1$ as above and unioning with
sets $X$ drawn uniformly from the region
$\{0\}\times [-0.8,0.8] \subset \mathbb{R}^2$ to form sets $\aS_2
= \aS_1 \cup X$.  When the added points are sufficiently numerous and
well-distributed, the point cloud now appears to have been sampled
from an underlying manifold homotopy equivalent to a figure 8 when the
scale parameter is sufficiently large.  (See Figure~\ref{fig:ann}.)
Computing the barcode for the first homology group now yields two long
intervals, displayed in Figure~\ref{fig2}.  (Again, this barcode was
stable over many repetitions.)

\begin{figure}[htp]
\centering
\includegraphics{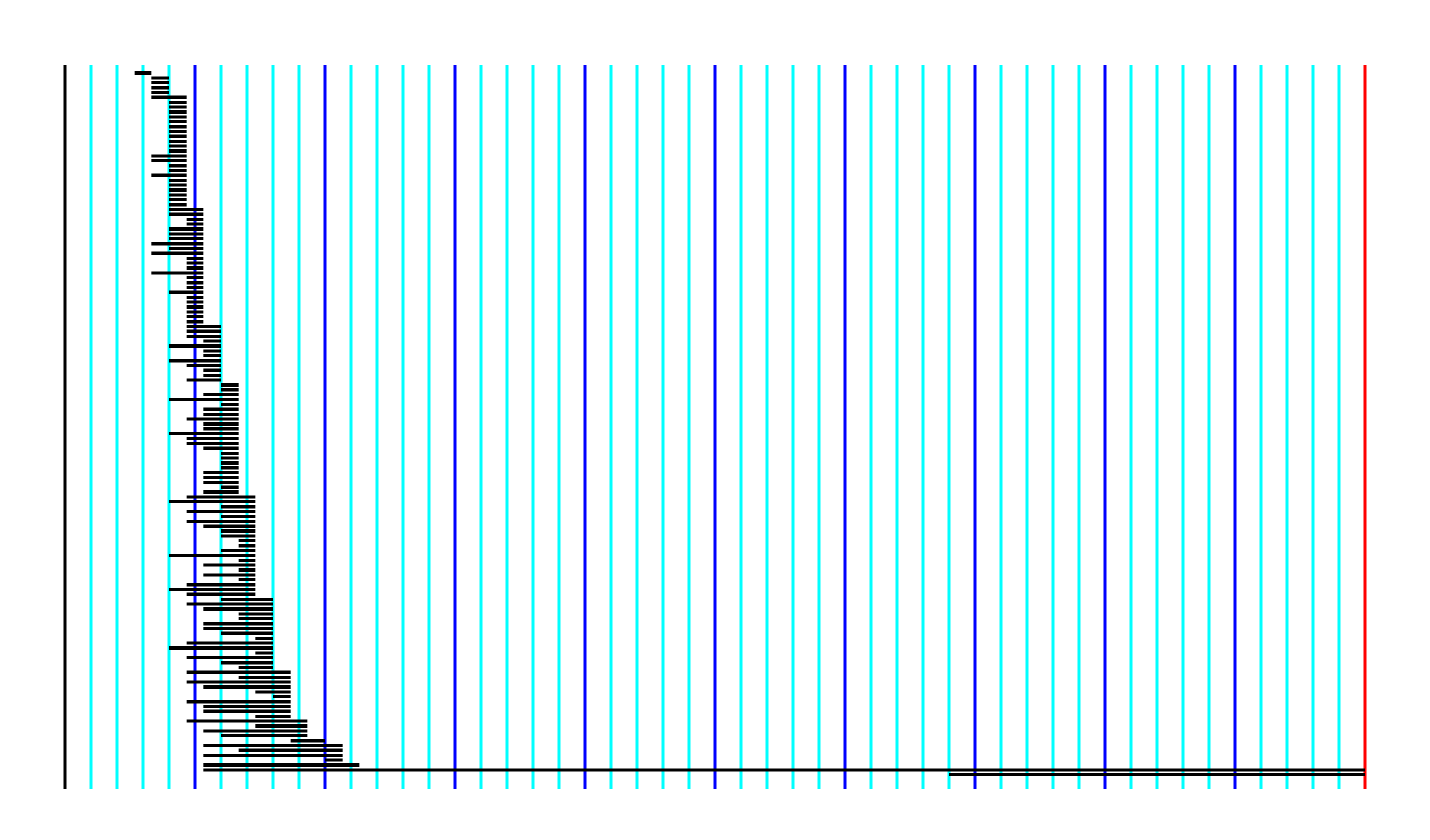}
\caption{Barcode for the annulus plus diameter linkage via the Vietoris-Rips
complex with 1000 points, showing 2 long bars.
Horizontal scale goes from 0 to 0.375. (Vertical scale is not meaningful.)}
\label{fig2}
\end{figure}

To test our methodology, we considered varying sizes for $X$ (as a
proportion of $|\aS_{1}|=1000$), and using
Algorithm~\ref{alg:relative} we computed $1000$ empirical
approximations to $\PHD_1^{75}(\aS_1)$ and $\PHD_1^{75}(\aS_2)$, using
the parameters $K = 1000$ and $n = 75$ and barcode cutoffs of $0.375$.

We then ran the following tests:

\begin{enumerate}

\item We compared the empirical distance distributions $\aD^{2}$ for $\aS_1$
and $\aS_2$ using the Kolmogorov-Smirnov statistic.

\item We computed histograms from $\aD^2$ for $\aS_1$ and $\aS_2$ (with
25 bins equally spaced over the maximum bounding region) and compared using the
$\chi^2$ test.

\item Fixing a reference barcode $B_{1}$ with a single long bar,
we computed the distance distribution $\aD_{B_1}$ for $\aS_1$ and
$\aS_2$ and repeated the comparisons above, using the Kolmogorov-Smirnov and
$\chi^2$ statistic (after forming histograms).  

\end{enumerate}

The results of these tests are summarized in Figure~\ref{fig:theta}.
We see that whereas the first two tests detect differences even with
relatively small amounts of malicious noise, the third test is less
sensitive and only begins to suggest rejection of the null hypothesis
around at $2.0\%$ or $2.5\%$ noise added.  (Note that in the third
test, the median of 
the distribution is precisely the statistic $\MHD_k^n$.)  On the one
hand, these results provide context for interpreting the results of
using the Kolmogorov-Smirnov and $\chi^{2}$ statistics with more
reasonable noise models (in other examples below).  On the other hand, we see that using
the third test we can extract robust topological information from the
data. 

\begin{figure}%
\begin{tabular}{|c|c|c|c|c|c|c|c|}
\hline
\headerstrut
Noise & $\chi^2$ 99\% & $\chi^2$ 95\% & $\chi^2$ 90\% &
KS 99\% & KS 95\% & KS 90\% \\
\hline
0.0\% & 0.0 & 0.0 & 0.0 & 0.0 & 0.0 & 0.0 \\ 
0.5\% & 0.05 & 0.05 & 0.05 & 0.2 & 0.2 & 0.2 \\
1.0\% & 0.05 & 0.15 & 0.15 & 0.2 & 0.45 & 0.55 \\
1.5\% & 0.15 & 0.2 & 0.35 & 0.25 & 0.4 & 0.65 \\
2.0\% & 0.2 & 0.45 & 0.55 & 0.35 & 0.5 & 0.65 \\
\hline
0.0\% & 0.0 & 0.0 & 0.0 & 0.0 & 0.0 & 0.0 \\
0.5\% & 0.0 & 0.0 & 0.0 & 0.0 & 0.0 & 0.0 \\ 
1.5\% & 0.0 & 0.0 & 0.05 & 0.0 & 0.05 & 0.1 \\ 
2.0\% & 0.0 & 0.1 & 0.15 & 0.0 & 0.1 & 0.2 \\ 
2.5\% & 0.1 & 0.15 & 0.2 & 0.35 & 0.55 & 0.65 \\ 
\hline
\end{tabular}
%
\caption{Comparison tests for samples $\aS_{1}$ and
$\aS_{2}=\aS_{1}\cup X$, where $\aS_{1}$ is a random sample of 1000
points from annulus $A$ and $X$ consists of a given proportion of
random ``noise'' points along the diameter.  Top: comparison tests for
the $\aD^2$ distribution (tests~(1) and~(2) in the text).  Bottom:
comparison tests for the $\aD(B_{1},-)$ distributions, where $B_{1}$
is the barcode with a single long bar (test~(3) in the text).}
\label{fig:theta}
\end{figure}

Finally, for a different application of the $\chi^{2}$ test to compare
these distributions, we used $k$-means clustering to produce discrete
distributions, as follows.  Performing $k$-means clustering on the
empirical approximations to $\PHD_k^n$ for $A$ indicated that the
resulting distributions had nontrivial mass clustered in three
regions: around a barcode $B_{0}$ with no long intervals, a barcode
$B_{1}$ with one long interval, and a barcode $B_{2}$ with two long
intervals.  This led to the following test, which we repeated $1000$
times.

\begin{enumerate}
\item Fixing $K = 1000$, for $\aS_1$ and $\aS_2$, we counted the number
of ``long bars'' (i.e., bars with length over a threshold of $0.125$,
which was determined by the $k$-means cluster centroids).

\item We use the $\chi^{2}$ test to determine if we can reject the
hypothesis that the resulting histograms were drawn from the same
distribution even at the 90\% level.
\end{enumerate}

The results were analogous to the more sensitive preceding
experiments; at $1.0\%$ noise added, we found that the $\chi^{2}$ test {\em
never} permitted rejection of the null hypothesis.  (As an example, a
sample distribution of masses on the centroid from a single run was
$0.017$, $0.983$, and $0$ for $\aS_1$ and $0.020$, $0.975$, and
$0.005$ for $\aS_2$.)  On the other hand, at $2.0\%$ noise added, we always
rejected the null hypothesis.  However, looking at the actual values,
we see that even at $5.0\%$ noise added, representative masses for
$\aS_2$ were $0.024$, $0.827$, and $0.149$.  We will see below how to
use confidence intervals to extract precise inferences about the
underlying homology from such data.

The example of the annulus also begins to illuminate a relationship
between the distributional invariants and density filtering.  Notice
that the second interval at the bottom of Figure~\ref{fig2} starts
somewhat later, reflecting a difference in average interpoint distance
between the original samples and the additional points added.  As a
consequence, one might imagine that appropriate density filtering
would also remove these points.  On the one hand, in many cases
density filtering is an excellent technique for concentrating on
regions of interest.  On the other hand, it is easy to construct
examples where density filtering fails --- for instance, we can build
examples akin to the one studied here where the ``connecting strip''
has comparable density to the rest of the annulus simply by reducing
the number of sampled points or by expanding the outer radius while
keeping the number of sampled points fixed.  In the former case our
methods also degrade, but the latter produces results akin to the
reported results above.  More generally, studying distributional
invariants (such as $\PHD^{n}_{k}$) by definition allows us to integrate
information from different density scales.  In practice, we expect
there to be a synergistic interaction between density filtering and
the use of $\PHD_{k}^{n}$; see Section~\ref{sec:notklein} for an
example of this interaction in practice.

\subsection*{Synthetic Example 2: Friendly circles}

Next, we considered a somewhat more complicated example.  The
underlying metric measure space $X$ is the subset of $\mathbb{R}^2$
specified as the union of the circle of radius $2$ centered at $(0,0)$
and the circle of radius $1$ centered at $(0.8,0)$, equipped with the
intrinsic metric and the length measure.  We sampled from $X$ by
choosing uniformly $\theta \in [0,2\pi]$ and assigning the indicated
point to the first circle with probability $\frac{2}{3}$ and the
second circle with probability $\frac{1}{3}$ (proportionally to their
lengths).  Convergence experiments 
analogous to those discussed in the previous example indicated
that choosing subsamples of cardinality greater than roughly $500$ resulted
in good approximations to $\PHD_k^n$.

Our experiments here are designed to indicate the robustness of our
invariants to both Gaussian and uniform noise --- the point of this
example is that noise points will introduce many classes in $H_1$ by
linking the two circles where they are near one another.  Once again,
it is illuminating to simply begin with persistent homology computed
from the entire subsample.  We sampled 1500 points from $X$.  We then
consider two noise models:
\begin{enumerate}
\item All points have ambient Gaussian noise added (i.e., we convolved
with a Gaussian of mean $0$ and covariance matrix $\sigma^{2}I_{2}$ in $\bR^2$).  (See
Figure~\ref{fig:twocirclesgaussian}.) 
\item A fraction of the points are replaced with uniform noise sampled
from the bounding rectangle $[-2,2] \times [-2,2] \subset \bR^2$. 
(See Figure~\ref{fig:twocircles}.)
\end{enumerate}

Computing the persistent homology from the Vietoris-Rips complex on these
points, without noise, we saw the expected pair of long bars in the
barcode for the first persistent homology group, computed using the
Vietoris-Rips complex.  With Gaussian noise, the results of computing the
barcodes degraded as the width of the Gaussians increased; for
example, when the width was $\sigma^{2}=0.1$ there were many long bars in the
barcode.  (We omit a graph of the barcode in the interest of space, as
the phenomenon is similar to the uniform noise case.)  And as uniform
noise was added, the results of computing barcodes using the Vietoris-Rips
complex degraded very rapidly, as we see in Figure~\ref{fig3} ---
there are many long bars.  This is precisely what one would expect in
light of the discussion in Section~\ref{sec:notrobust} and the
geometry of the situation.

\begin{figure}[htp]
\begin{center}
\includegraphics{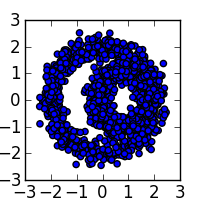}
\end{center}
\caption{Two circles with Gaussian noise.}
\label{fig:twocirclesgaussian}
\end{figure}

\begin{figure}[htp]
\begin{center}
\includegraphics{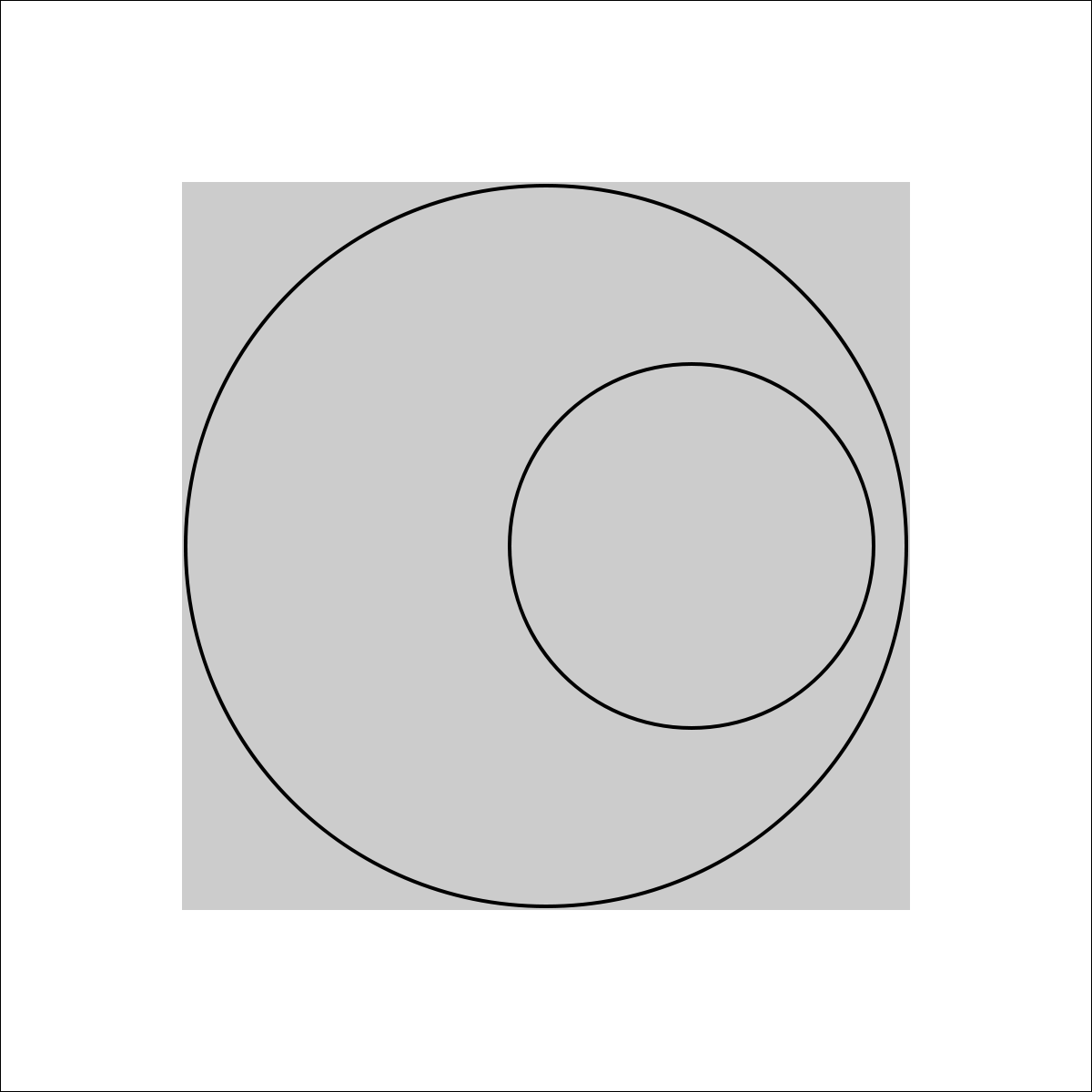}
\end{center}
\caption{Two circles with uniform noise (indicated by gray box).}
\label{fig:twocircles}
\end{figure}

\begin{figure}[htp]
\begin{center}
\includegraphics{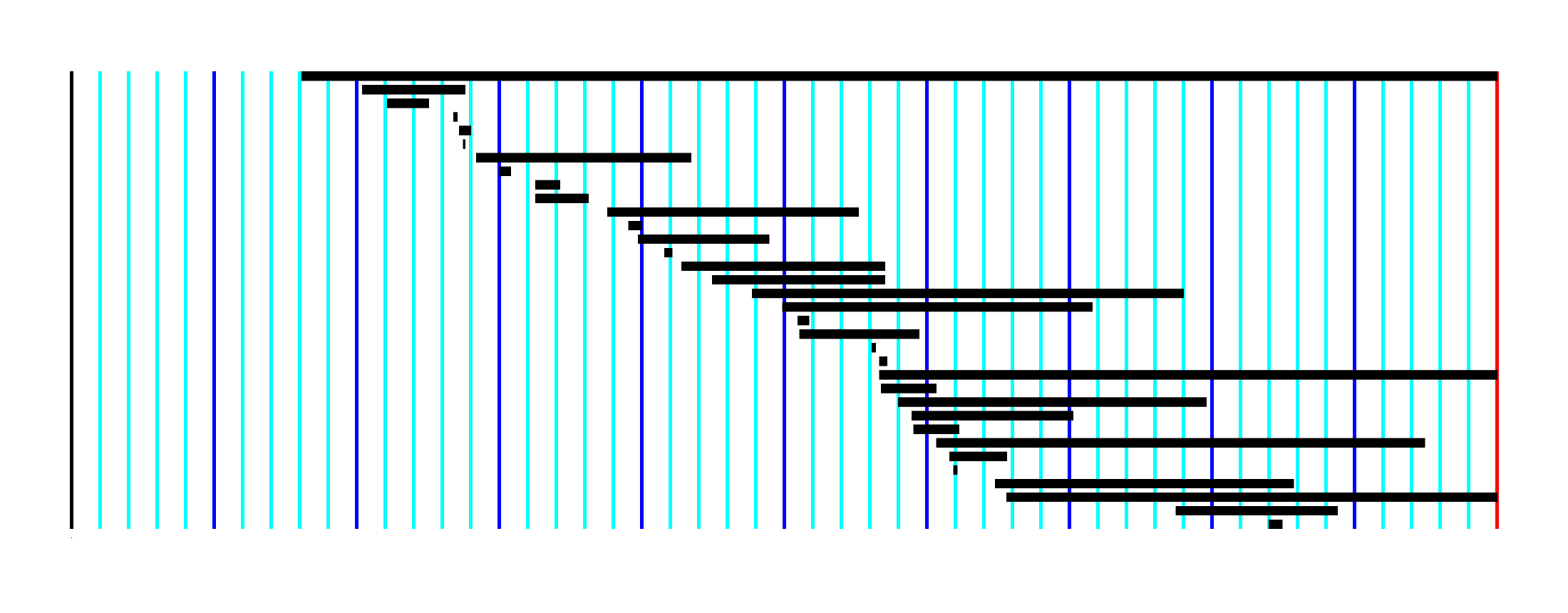}\\
\includegraphics{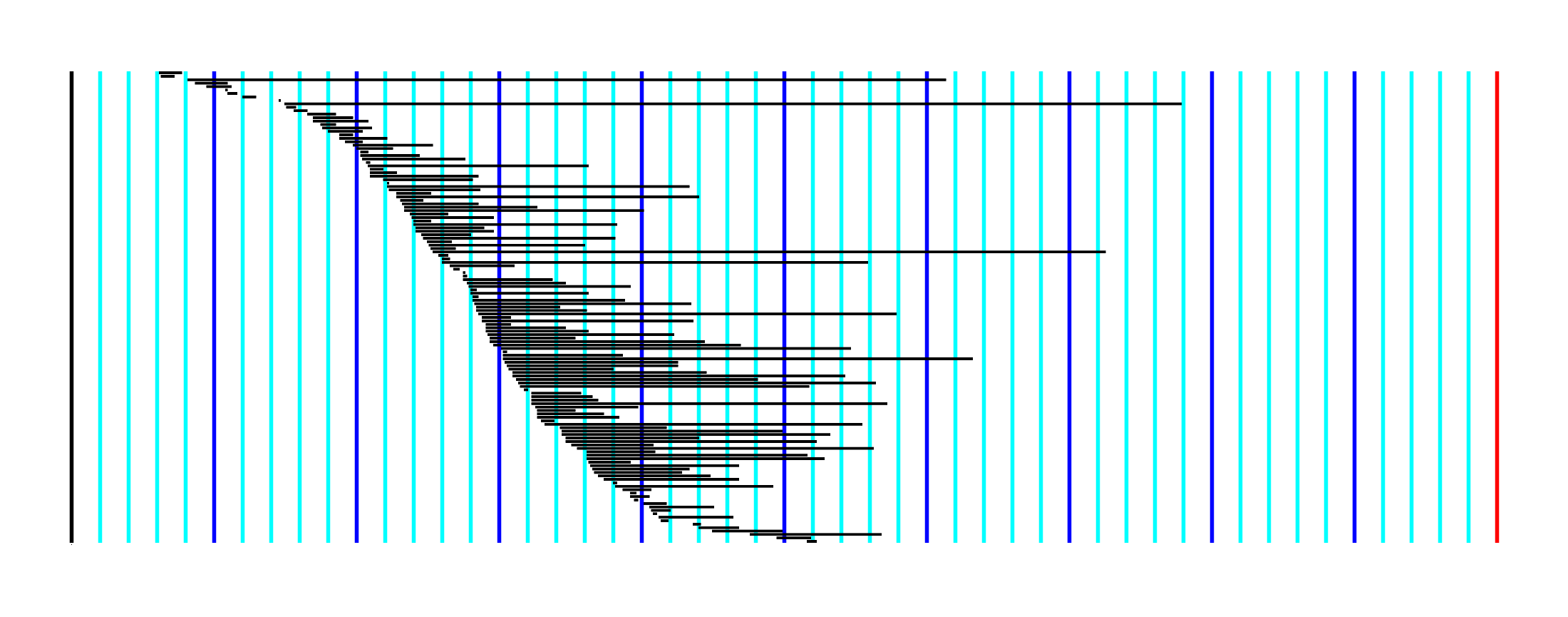}\\
\includegraphics{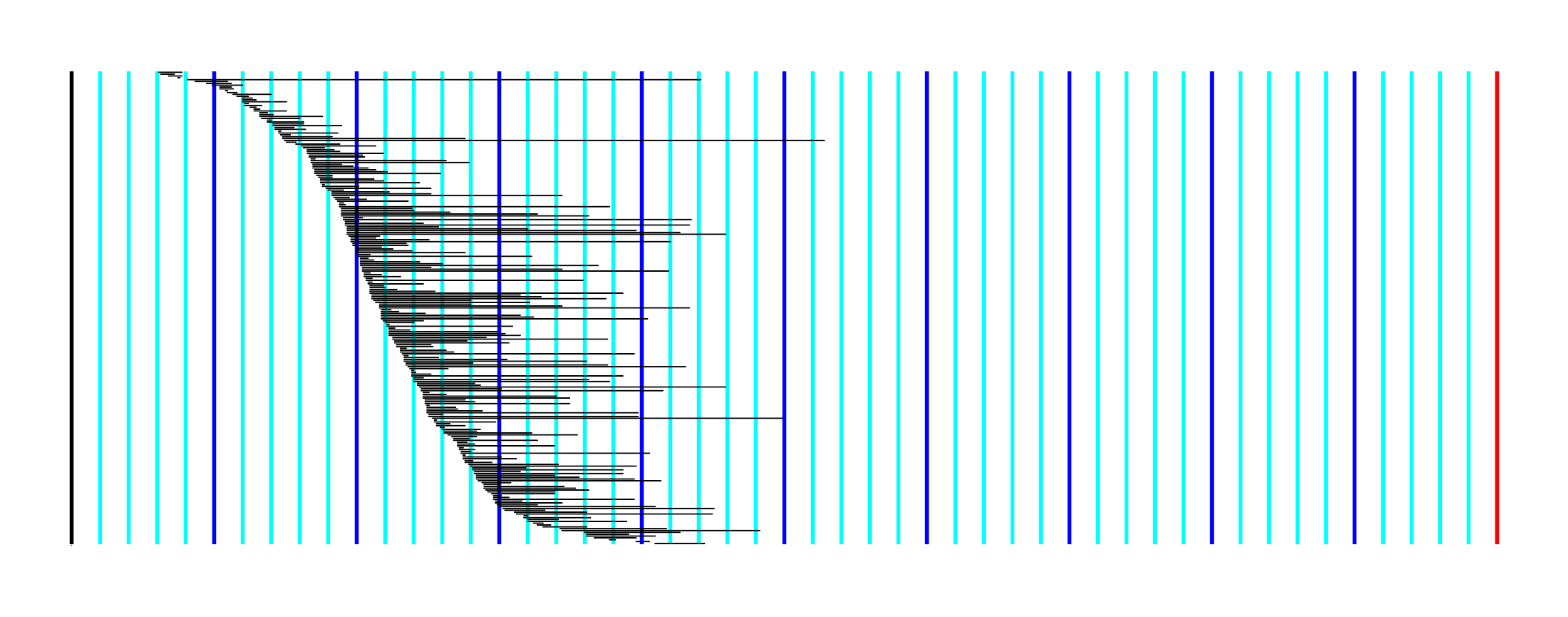} 
\end{center}
\caption{Barcode for two circles with 10, 50, and 90 noise points.
Horizontal scale goes from 0 to 0.75. (Vertical scale is not meaningful.)}
\label{fig3}
\end{figure}

Even with only 10 noise points, we see 3 bars, and with 90 noise
points there are 12.  (These results were stable across different
samples; we report results for a representative run.)

In contrast, we computed $\PHD_1^{300}$ for the same point clouds
(i.e., the two circles plus varying numbers of noise points), using $K
= 1000$ samples of size $300$ and a cutoff of $0.75$.  The resulting
empirical distributions had essentially all of their weight
concentrated around barcodes with a small number of long intervals
(revealed once again by $k$-means clustering).  For the points in the
empirical estimate of $\PHD_1^{300}$ around we counted the number of
``long bars'' with length over the threshold of $0.25$ (again
determined from the $k$-means centroids).  The results are summarized
in Figure~\ref{table2} below.

\begin{figure}[htp]
{\centering
\begin{tabular}{c c c c c c c}
\hline\hline
\headerstrut
Number of noise pts & 0 bars & 1 bars & 2 bars & 3 bars & 4 bars & 5
bars \\ [0.125ex]
\hline
0 & 0 & 303 & 696 & 1 & 0 & 0 \\
10 & 0 & 305 & 589 & 106 & 0 & 0 \\
20 & 0 & 278 & 590 & 132 & 0 & 0 \\
30 & 0 & 285 & 594 & 119 & 2 & 0 \\
40 & 1 & 259 & 584 & 149 & 6 & 1 \\
50 & 0 & 289 & 553 & 154 & 4 & 0 \\
60 & 0 & 254 & 591 & 146 & 7 & 2 \\
70 & 0 & 277 & 564 & 154 & 5 & 0 \\
80 & 1 & 229 & 543 & 196 & 29 & 2 \\
90 & 0 & 229 & 533 & 207 & 28 & 3 \\
\hline
\end{tabular}}
\caption{Distribution summaries for $\PHD_1^{300}$ in ``Friendly
Circles'' example for number of long
bars occurring in $1000$ tests with given number of noise points
added.}
\label{table2}
\end{figure}

A glance at the table shows that the majority of the weight is
clustered around a barcode with 2 long bars and that the data
overwhelming supports a hypothesis of $\leq 3$ barcodes under all
noise regimes.  More precisely, the likelihood statistic of
Section~\ref{sec:like} allows us to evaluate the hypothesis $H$ that
the observed empirical approximation to $\PHD_{k}^{n}$ was drawn from
an underlying barcode distribution with weight $\geq 5\%$ on barcodes
with more than 3 long bars.  In the strictest tests with 80 and 90
noise points, 31 out of 1000 samples were near barcodes with more than
3 long bars, and so we estimate that the probability of the
distribution having $\geq 5\%$ of the mass at 4 or more barcodes as
$\leq \text{BD}(1000,31,.05)<0.22\%$.  Put another way, we can reject
the hypothesis that the actual distribution has more than $5\%$ mass
at 4 or more barcodes at the $99.7\%$ level.

We also ran a similar experiment with Gaussian noise, looking at
$\Phi_1^{300}$ and varying widths; the results are summarized in the
Figure~\ref{table:gaussnoise} below.  As one would expect,
sufficiently wide Gaussians cause
the smaller circle to appear to be a (contractible) disk attached to
the larger circle.

\begin{figure}[htp]
{\centering
\begin{tabular}{c c c c c c c}
\hline\hline
\headerstrut\relax
$\sigma^{2}$& 0 bars & 1 bars & 2 bars & 3 bars & 4 bars & 5
bars \\[0.125ex]
\hline
0.05 & 2 & 59 & 930 & 9 & 0 & 0 \\
0.075 & 44 & 351 & 585 & 20 & 0 & 0 \\
0.1 & 204 & 537 & 249 & 10 & 0 & 0 \\
\hline
\end{tabular}}
\caption{Distribution summaries for $\PHD_1^{300}$ in ``Friendly
Circles'' example for number of long
bars occurring in $1000$ tests with Gaussian noise added of mean 0 and
covariance $\sigma^{2}I_{2}$.}
\label{table:gaussnoise}
\end{figure}

\subsection*{Spheres and tori in $\mathbb{R}^3$}

We now turn to more realistic synthetic examples that are less easily
summarized (and better represent the ambiguity present in the typical
application of topological data analysis).  We studied two standard
geometric examples of smooth manifolds.

\begin{enumerate}
\item Two-dimensional spheres of varying radii $r$, which we denote $S(r)$,
\item Tori of inner radius $r$ and outer radius $R$ for varying parameter
values which we denote $T(r,R)$ (see figure~\ref{fig:torus}).
\end{enumerate}

\begin{figure}[htp]
\begin{center}
\includegraphics{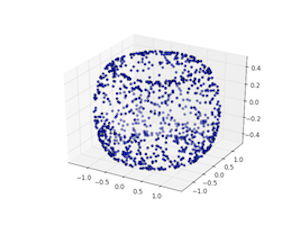}
\end{center}
\caption{Torus $T(0.5,1)$}
\label{fig:torus}
\end{figure}

These examples have interestingly different characteristics; detecting
the sphere's top homology class is relatively easy even in the face of
noise, whereas noise can introduce many spurious homology classes in
degree $1$.  In contrast, the torus $T(0.5, 1)$ is sufficiently
different in the scale of its two axes that detecting the degree $2$
homology class and both degree $1$ homology classes is quite
challenging.

There are various reasonable choices to make about how to sample from
these objects.  In our experiments, we use the intrinsic metric and
sample using the area measure in each case:

\begin{enumerate}

\item To draw a uniform point on the sphere using the area measure, we
draw points $z_1, z_2, z_3$ from the standard normal distribution and
consider the point $(\frac{z_1}{\sqrt{z_1^2 + z_2^2 +
z_3^2}}, \frac{z_2}{\sqrt{z_1^2 + z_2^2 +
z_3^2}}, \frac{z_3}{\sqrt{z_1^2 + z_2^2 + z_3^2}})$.

\item To draw a uniform point on the torus using the area measure, we
parametrize the torus as 
\[
(\theta, \psi) \mapsto \left( (R + r\cos(\theta))\cos(\psi), (R +
r\cos(\theta))\sin(\psi), r\sin(\theta) \right),
\]
for $0 \leq \theta, \psi \leq 2\pi$ and use the rejection sampling
procedure explained in~\cite[2.2]{diaconis}.  (Note that drawing
$\theta$ and $\psi$ uniformly in $[0,2\pi]$ does not work.)

\end{enumerate}

We again work with two noise models, adding both Gaussian noise (by
convolving with a mean $0$ Gaussian with covariance matrix
$\sigma^{2}I_{3}$ in $\mathbb{R}^3$) to all points and replacing some
of the points with uniform noise (obtained from uniform samples in
$\mathbb{R}^3$ using the bounding box $[-2,2]^{3}$) to the samples.
We note that these two noise models are somewhat different in
character; the Gaussian noise affects all points, whereas the uniform
noise corrupts some fraction of the total number of points.

Our first set of experiments studied the rate of convergence in
Corollary~\ref{cor:largenum}; our methodology is the same as in the
previous section, and we find that acceptable minimum cardinalities
for $S \subset X$ in order for $\Phi^n_k(S)$ and $\Phi^n_k(X)$ (for
varying $n$) to be indistinguishable to the $\chi^2$ and
Kolmogorov-Smirnov tests are around $1000$ for the sphere and $2000$
for the torus.  We fix $K = 1000$ throughout.  We use these results as
a guide when carrying out experiments analyzing the metric measure
spaces in this region.

Next, in order to explore how the inference procedures described in
Section~\ref{sec:intervals} can be used in the context of our running
examples in $\mathbb{R}^3$, we carry out the following different
experiments, again using Algorithms~\ref{alg:absolute}
and~\ref{alg:relative} as the base.

\begin{enumerate}
\item We use the Kolmogorov-Smirnov and $\chi^2$ tests to study how
the distribution $\aD^2$ of distances induced from $\PHD$ changes as
noise is added.

\item We use estimates of $\MHD_k^n$ both to extract information about
the salient topological features of the sphere and the torus and also
to test the robustness of this invariant to added noise.
\end{enumerate}

We began by looking at what the Kolmogorov-Smirnov and $\chi^2$ tests
tell us about the sphere and torus.  Working with the uniform noise
model, we used subsamples of $1000$ points for the sphere and
considered $\PHD_2^{150}$ as our base.  For the torus, we used
subsamples of $2000$ points and $\PHD_1^{150}$ for our base.  We
replaced an increasing fraction of the points with noise and compared
to the distribution from the underlying (noiseless) model, with the
results summarized in Figure~\ref{fig:spherenoise}.  Here the
percentage in the table once again indicates the fraction of runs in
which we could reject the null hypothesis of the same distribution at
the indicated significance level.  A clear conclusion to draw is that
$\PHD_k^n$ is relatively insensitive to even large amounts of uniform
noise.  In contrast, when the corresponding experiments were run with
the Gaussian noise model, we found that there was a threshold effect;
for noise widths smaller than roughly $\sigma^{2}=0.05$, the distributions could
not be distinguished by these tests, but for larger noise widths they
basically always appeared to be distinct.

\begin{figure}
\begin{tabular}{|c|c|c|c|c|c|c|c|c|}
\hline\headerstrut
Shape & Noise& $\chi^2$ 99\% & $\chi^2$ 95\% & $\chi^2$ 90\% &
KS 99\% & KS 95\% & KS 90\% \\
\hline
S(1) & 1\% & 0.0 & 0.0 & 0.0 & 0.0 & 0.05 & 0.05 \\ 
S(1) & 5\% & 0.0 & 0.0 & 0.0 & 0.0 & 0.0 & 0.05 \\ 
S(1) & 10\% & 0.0 & 0.0 & 0.1 & 0.0 & 0.1 & 0.1 \\ 
S(1) & 20\% & 0.0 & 0.1 & 0.15 & 0.1 & 0.2 & 0.35 \\ 
\hline
T(0.5,1) & 1\% & 0.0 & 0.0 & 0.0 & 0.0 & 0.0 & 0.0 \\ 
T(0.5,1) & 5\% & 0.0 & 0.0 & 0.0 & 0.0 & 0.0 & 0.05 \\ 
T(0.5,1) & 10\% & 0.0 & 0.05 & 0.05 & 0.0 & 0.05 & 0.1 \\ 
T(0.5,1) & 20\% & 0.0 & 0.1 & 0.2 & 0.15 & 0.2 & 0.3 \\ 
\hline
\end{tabular}
\vspace{2 pt}
\caption{Comparisons tests for the sphere (for Betti 2) and torus
(for Betti 1) for distributions with
uniform noise replacing a fraction of the points compared against the
noiseless distributions.}\label{fig:spherenoise}  
\end{figure}

Although the previous experiments indicate the degree to which
$\PHD_k^n$ is robust against noise, in practice it is more likely that
we will want to extract information about easily expressed hypotheses
concerning the rank of the homology groups of the underlying space.
To this end, we consider $\MHD_k^n$ with regards to various reference
barcodes; let $m[a,b)$ denote the barcode consisting of $m$
copies of the interval $[a,b)$.  We used Algorithms~\ref{alg:absolute}
and ~\ref{alg:relative} to compute $\MHD_k^n$ and we used the
asymptotic estimates of Definition~\ref{defn:confint} to produce the
confidence intervals.  We chose subsets of size $1000$ to subsample
from.

We begin by considering results for uniform noise in a box, as a
reference benchmark; the results are summarized in
Figure~\ref{fig:MHDnoiseints}.  We then compute for the sphere; the
results are summarized in Figure~\ref{fig:MHDsphereints} below.
Finally, we did the computations for the torus; the results are
summarized in Figure~\ref{fig:MHDtorusints} below.  We obtained the
reference barcodes by inspection of a single run; this procedure is a
proxy for the kind of exploratory data analysis that we expect would
generate the hypotheses to test using our test statistics.  The
confidence intervals in the table were generated by using $100$
samples; the reported results are representative for these parameter
settings.  We also ran a number of experiments with Gaussian noise
as well.  In the interest of space, we report only the results on the
torus, which are summarized in Figure~\ref{fig:MHDtorusintsg}, as these
are representative.

Before we begin to discuss these results, a few observations about the
data sets are in order.  We expect that the sphere should be a
relatively easy example; uniform noise is unlikely to interfere with
the top-dimensional homology class.  This expectation is borne out by
simply computing the persistent homology using 1000 points --- even
with 10\% uniform noise added, we see a single much longer bar.  (We
omit the picture of this.)  In this situation, we regard our
experiments as validating the use of $\MHD_k^n$ to make precise
statistical statements about topological hypotheses.  In contrast, the
torus $T(0.5,1)$ is a difficult test; the scale of the two one-cycles
is different, and we need a large number of points in order to resolve
them both.  When running the persistent homology using all 1000
points, even tiny amounts of uniform noise cause substantial
disruptions in the results, i.e., many long bars.  (We again omit the
picture of this.)  As a consequence, in the presence of noise, working
without the statistical methodology makes it basically impossible to draw
conclusions about the data.

For the sphere, the measured results indicate that $\MHD_2^{150}$ does
an excellent job of detecting the class in dimension $2$.
Specifically, until the noise reaches 20\%, the confidence interval
for the hypothesis $1[0.4,0.55)$ is the closest to $0$ and does not
overlap with the other confidence intervals.  When confidence
intervals for different population quantities do not overlap, the
difference between the two is statistically significant at the 99\%
level.  We could also use Monte Carlo simulation to estimate the
difference between the medians (for the two hypotheses), if a more
refined test statistic explicitly comparing the hypotheses was
desired.  The measured results do not detect any classes in dimension
$1$, even with really substantial amounts of noise.  (The results are
comparable to the results for the box in dimension $1$.)

For the torus, we begin by discussing the case of uniform noise.  In
dimension $1$ we see that both $1$ and $2$ bar variants are close to
the observed data.  When we perform Monte Carlo simulation to obtain
confidence intervals for the difference between the medians, the 95\%
intervals contain $0$ --- this suggests that we cannot distinguish
between the two hypotheses with this test statistic.  One
interpretation of this result is that there are in fact a larger
number of long bars, and indeed inspection of the barcode results
reflect approximately $5$ ``long'' bars.  It is encouraging that our
results are very robust in the face of large amounts of uniform noise,
however.  We can obtain better results by increasing the number of
samples points; when using $\MHD_1^{500}$ and a subsample of size
$10000$, the medians and confidence intervals for $2$ bars is
substantially smaller than for $1$ bar or $3$ bars --- the
difference is now statistically significant at the 99\% level.
(For reasons of space we omit reporting the specific tables.)

For the torus with Gaussian noise, the results admit a comparable
analysis, with the exception of the fact that we see a substantial
degradation as the width increases (and at noise of width $\sigma^{2}=0.1$ our
procedures are basically useless).

\begin{figure}
\begin{tabular}{|c|c|c|c|}
\hline\headerstrut
$k$ & $m$ & Median & 95\% Confidence interval \\
\hline
1 & 0 & 0.09 & [0.0875,0.0925] \\
1 & 1 & 0.08 & [0.0775,0.0825] \\
1 & 2 & 0.075 & [0.075,0.0775] \\
\hline
2 & 0 & 0.0175 & [0.0125,0.025] \\
2 & 1 & 0.085 & [0.075,0.1] \\
2 & 2 & 0.095 & [0.085,0.115] \\
\hline
\end{tabular}
\vspace{2 pt}
\caption{Confidence intervals for $\MHD_k^{150}$ applied to uniform noise in $[-2,2]
\times [-2,2]$ for reference bar codes $m[0.40.55)$.}\label{fig:MHDnoiseints} 
\end{figure}

\begin{figure}
\begin{tabular}{|c|c|c|c|c|c|}
\hline\headerstrut
Shape & Noise& $k$ & $m$& Median & 95\% Confidence interval \\
\hline
S(1) & 0\% & 1 & 0 & 0.0925 & [0.09,0.0975] \\
S(1) & 0\% & 1 & 1 & 0.195 & [0.195,0.2] \\
S(1) & 0\% & 1 & 2 & 0.205 & [0.205,0.21] \\
\hline
S(1) & 5\% & 1 & 0 & 0.095 & [0.0925,0.1] \\
S(1) & 5\% & 1 & 1 & 0.175 & [0.17,0.175] \\
S(1) & 5\% & 1 & 2 & 0.185 & [0.185,0.19] \\
\hline
S(1) & 10\% & 1 & 0 & 0.0975 & [0.095,0.1025] \\
S(1) & 10\% & 1 & 1 & 0.15 & [0.14,0.165] \\
S(1) & 10\% & 1 & 2 & 0.18 & [0.175,0.185] \\
\hline
S(1) & 20\% & 1 & 0 & 0.0975 & [0.0925,0.1025] \\
S(1) & 20\% & 1 & 1 & 0.115 & [0.105,0.12] \\
S(1) & 20\% & 1 & 2 & 0.145 & [0.135,0.155] \\
\hline
\hline 
S(1) & 0\% & 2 & 0 & 0.07 & [0.065, 0.075] \\
S(1) & 0\% & 2 & 1 & 0.02 & [0.02, 0.025] \\
S(1) & 0\% & 2 & 2 & 0.075 & [0.075, 0.075] \\
\hline
S(1) & 5\% & 2 & 0 & 0.065 & [0.0625, 0.07] \\ 
S(1) & 5\% & 2 & 1 & 0.025 & [0.015, 0.03] \\
S(1) & 5\% & 2 & 2 & 0.075 & [0.075, 0.075] \\
\hline
S(1) & 10\% & 2 & 0 & 0.06 & [0.0575,0.065] \\
S(1) & 10\% & 2 & 1 & 0.03 & [0.025,0.035] \\
S(1) & 10\% & 2 & 2 & 0.075 & [0.075,0.08] \\
\hline
S(1) & 20\% & 2 & 0 & 0.0525 & [0.045,0.0575] \\
S(1) & 20\% & 2 & 1 & 0.045 & [0.04,0.065] \\
S(1) & 20\% & 2 & 2 & 0.075 & [0.075,0.075] \\
\hline
\end{tabular}
\caption{Confidence intervals for $\MHD_k^{150}$ for reference barcodes
$m[0.4,0.55)$
applied to the sphere with a given percentage of points replaced with
uniform noise.}\label{fig:MHDsphereints} 
\end{figure}

\begin{figure}
\begin{tabular}{|c|c|c|c|c|c|}
\hline\headerstrut
Shape & Noise& $k$ &$m$ & Median & 95\% Confidence interval \\
\hline
T(0.5,1) & 0\% & 1 & 0 & 0.1575 & [0.1525,0.16] \\
T(0.5,1) & 0\% & 1 & 1 & 0.0925 & [0.09,0.0975] \\
T(0.5,1) & 0\% & 1 & 2 & 0.1 & [0.0975,0.1] \\
\hline
T(0.5,1) & 5\% & 1 & 0 & 0.15 & [0.145,0.1525] \\
T(0.5,1) & 5\% & 1 & 1 & 0.095 & [0.0925,0.0975] \\
T(0.5,1) & 5\% & 1 & 2 & 0.0975 & [0.095,0.1] \\
\hline
T(0.5,1) & 10\% & 1 & 0 & 0.145 & [0.14,0.15] \\
T(0.5,1) & 10\% & 1 & 1 & 0.0975 & [0.0925,0.1025] \\ 
T(0.5,1) & 10\% & 1 & 2 & 0.0975 & [0.095,0.1] \\
\hline
T(0.5,1) & 20\% & 1 & 0 & 0.14 & [0.1375,0.1425] \\
T(0.5,1) & 20\% & 1 & 1 & 0.0925 & [0.09,0.095] \\
T(0.5,1) & 20\% & 1 & 2 & 0.0975 & [0.0925,0.1] \\
\hline
\end{tabular}
\vspace{2 pt}
\caption{Confidence intervals for $\MHD_k^{150}$ for reference barcodes
$m[0.3,0.55)$
applied to the torus with a given percentage of points replaced with
uniform noise.}\label{fig:MHDtorusints}
\vspace{2em}
\end{figure}

\begin{figure}
\begin{tabular}{|c|c|c|c|c|c|}
\hline\headerstrut
Shape & $\sigma^{2}$ & $k$ & $m$ & Median & 95\% Confidence interval \\
\hline
T(0.5,1) & 0.01 & 1 & 0 & 0.145 & [0.1425,0.15] \\
T(0.5,1) & 0.01 & 1 & 1 & 0.085 & [0.0825,0.0875] \\
T(0.5,1) & 0.01 & 1 & 2 & 0.115 & [0.11,0.125] \\
\hline
T(0.5,1) & 0.05 & 1 & 0 & 0.1275 & [0.1225,0.135] \\
T(0.5,1) & 0.05 & 1 & 1 & 0.0775 & [0.075,0.0825] \\ 
T(0.5,1) & 0.05 & 1 & 2 & 0.11 & [0.105,0.115] \\
\hline
T(0.5,1) & 0.1 & 1 & 0 & 0.095 & [0.0925,0.1] \\
T(0.5,1) & 0.1 & 1 & 1 & 0.0875 & [0.085,0.0925] \\ 
T(0.5,1) & 0.1 & 1 & 2 & 0.105 & [0.105,0.11] \\
\hline
\end{tabular}
\caption{Confidence intervals for $\MHD_k^{150}$ for reference barcodes
$m[0.3,0.55)$
applied to the torus with
Gaussian noise of mean~$0$ and covariance matrix $\sigma^{2}I_{3}$.}\label{fig:MHDtorusintsg} 
\end{figure}

\section{Application: confidence intervals for the natural images
dataset}\label{sec:notklein}

One of the most prominent applications of persistent homology in
topological data analysis is the study of the natural images dataset
described in~\cite{natimages}.  This data consists of $3 \times 3$
patches sampled from still photographs of ``natural'' scenes (i.e.,
pictures of rural areas without human artifacts).  The results of
Carlsson, Ishkhanov, de Silva, and
Zomorodian~\cite{carlssonetalvision} extract topological signals from
this data set which can be interpreted in terms of collections of
patches which are known to be meaningful based on the neurophysiology
of the eye.  The goal of this section is to apply our statistical
methodology to validate the conclusions of their work.

\subsection{Setup}

We compute the confidence intervals based on $\MHD_k^n$ for a subset
of patches from the natural images dataset as described
in~\cite{carlssonetalvision}.  We briefly review the setup.  The
dataset consists of 15000 points in $\bR^8$, generated as follows.
From the natural images, $3 \times 3$ patches (dimensions given in
pixels) were sampled and the top $30\%$ with the highest contrast
were retained.  These patches were then normalized twice, first by
subtracting the mean intensity and then scaling so that the Euclidean
norm is $1$.  The resulting dataset can be regarded as living on the
surface of an $S^7$ embedded in $\bR^8$.  After performing density
filtering (with a parameter value of $k=15$; refer
to \cite{carlssonetalvision} for details) and randomly selecting
$15000$ points, we are left with the dataset $\aM(15,30)$.  At this
density, one tends to see a barcode corresponding to 5 cycles in the
$H_1$. In the Klein bottle model, these cycles are generated by three
circles, intersecting pairwise at two points (which can be visualized
as unit circles lying on the $xy$-plane, the $yz$-plane, and the
$xz$-plane).

\subsection{Results}

We computed empirical approximations to $\PHD_1^{500}(\aM(15,30))$
using Algorithm~\ref{alg:absolute}, with $K = 1000$ and using a
barcode cutoff of $2$ (we used the value reported
in~\cite{carlssonetalvision} as the maximal filtration value). 
We found that (after applying $k$-means clustering, as above) the
weight was distributed as 0.1\% with one long bar, 1.1\% with two
long bars, 7.4\% with three long bars, 34.2\% with four long bars, and
57.2\% with five long bars.  (Here the threshold for a long bar was
$1$.)  Analyzing likelihood test statistics as in
Section~\ref{sec:like}, we find that the underlying distribution has
at least 95\% of its mass on two, three, or four bars at the the
99.7\% confidence level.

We also analyze the results using $\MHD$.  We use as the hypothesis
barcode the multi-set $5[0,2)=\{[0,2),[0,2),[0,2),[0,2),[0,2)\}$.  We find
using the nonparametric estimate from Definition~\ref{defn:confint}
that the 95\% confidence interval for $\MHD_1^{500}(\aM(15,30)$ is
$[0.442, 0.476]$.  The 99\% confidence interval for
$\MHD_1^{500}(\aM(15,30))$ is $[0.436, 0.481]$.  These results
represents high confidence for the data to be further than $0.442$ but
closer than $0.476$ to the reference barcode.  On the other hand, when
we compute the confidence intervals using the reference barcode the
empty set, we find that both endpoints for the 95\% and 99\%
confidence intervals are the cutoff value of $2$.  We find the same
results for hypothesis barcodes with $\ell$ bars $[0,2)$ for
$0 \leq \ell \leq 10, \ell \neq 5$.  In particular, this means that
the differences between the distance to the $5$ bar hypothesis and any
other is statistically significant at the 99\% level.

We interpret these results to suggest that the hypothesis barcode is
consistent with the underlying distribution amongst barcode
distributions that put all of their mass on a single barcode.  Of
course, these results also suggest that when sampling at 500 points,
we simply do not expect to see a distribution that is heavily
concentrated around a single barcode.  In the next subsection, we
discuss the use of the witness complex, which does result in such a
narrow distribution.

\begin{rem}
To validate the non-parametric estimate of the confidence interval, we
also used bootstrap resampling to compute bootstrap confidence
intervals.  Although we do not justify or discuss further this
procedure herein, we note that we observed
the reassuring phenomenon that the bootstrap confidence intervals
agreed closely with the non-parametric estimates for both the 95\%
confidence intervals and the 99\% confidence intervals in each
instance.
\end{rem}

\subsection{Results with the witness complex}

Because of the size of the datasets involved, the analysis
performed in~\cite{carlssonetalvision} used the weak witness complex
$\W$ rather than the Vietoris-Rips 
complex $\VR$.  The weak witness
complex for a metric space $(X, \partial)$ depends on a subset
$X_{0}\subset X$ of witnesses; the size of the complexes is controlled by
$|X_{0}|$ and not $|X|$.

\begin{defn}
For $\epsilon \in \bR$, $\epsilon \geq 0$ and witness set $X_{0} \subset
X$, the weak witness complex $\W_\epsilon(X,X_{0})$ is the simplicial
complex with vertex set $X_{0}$ such 
that $[v_0, v_1, \ldots, v_n]$ is an $n$-simplex when for each pair
$v_i, v_j$, there exists a point $p \in X$ (a witness) such that the
distances $\partial(v_i, p) \leq \epsilon$.
\end{defn}

When working with the witness complex, we adapt our basic approach to
study the induced distribution on barcodes which comes from fixing the
point cloud and repeatedly sampling a fixed number of witnesses.  The
theoretical guarantees we obtained for the Vietoris-Rips complex in this paper
do not apply directly; we intend to study the robustness and
asymptotic behavior of this process in future work. Here, we report
preliminary numerical results.

Specifically, we again computed empirical approximations to
$\PHD_1^{n}(\aM(15,30))$ using Algorithm~\ref{alg:absolute}, with $K =
1000$ and using a barcode cutoff of $2$.  However, to produce the
underlying complex, we use the $n$ points for each subsample as the
landmark points $X_{0}$ in the construction of the witness complex
rather than as the vertices for the Vietoris-Rips complex.

We use as the hypothesis barcode the multi-set $5[0,2)$ as above.  We found using
the non-parametric estimate of Definition~\ref{defn:confint} that the
95\% confidence interval for $\MHD_1^{100}(\aM(15,30))$, is $[0.024,
0.027]$.  The 99\% confidence
interval for $\MHD_1^{100}(\aM(15,30))$ was also $[0.024, 0.027]$.
When we computed the 95\% confidence interval for
$\MHD_1^{150}(\aM(15,30)$ we obtained $[0.021, 0.023]$.  The 99\%
confidence interval for $\MHD_1^{150}(\aM(15,30))$ was $[0.021,
0.024]$.  This represents high confidence for the data to be further
than $0.021$ (for $\PHD_{1}^{150}$) and $0.024$ (for $\PHD_{1}^{100}$)
but closer than $0.024$ (for $\PHD_{1}^{150}$) and $0.027$ (for
$\PHD_{1}^{100}$) to the reference barcode.  We obtained essentially
the same results $\MHD^{500}_{1}$ as for $\MHD^{150}_{1}$.  On the
other hand, when using hypothesis barcodes with $\ell$ bars $[0,2)$
for $0 \leq \ell \leq 10, \ell \neq 5$, the confidence intervals
start and end at $2$.  Again, this means that the difference between
the distances to the 5-bar hypothesis and the other hypotheses is
statistically significant at the 99\% level.
We interpret these results to mean that the
underlying distribution is essentially concentrated around the
hypothesis barcode; the distance of $0.025$ is essentially a
consequence of noise.

\begin{rem}
In contrast, when we compute $\MHD_1^{25}(\aM(15,30))$ (using the same
experimental procedure as above), we find the confidence interval is
$[1.931, 1.939]$.  When we compute $\MHD_1^{75}(\aM(15,30))$, we find
that the confidence interval is $[1.859, 1.866]$.  This represents
high confidence that $\MHD_{1}^{25}$ and $\MHD_{1}^{75}$ are far from
this reference barcode, which in light of the confidence intervals
above for $\MHD^{150}_{1}$ and $\MHD^{500}_{1}$ appear to indicate
that samples sizes 25 and 75 are too small.
\end{rem}

\bibliographystyle{alpha}

\end{document}